\documentclass[letter]{article} 
\usepackage{amsmath,amssymb,amsthm}
\usepackage[colorlinks=true,citecolor=black,linkcolor=black,urlcolor=blue]{hyperref}

\title{Balanced Allocation on Hypergraphs\thanks{The first author is supported by the Australian Research Council
		Discovery Project DP190100977. The second and third authors are supported by the Australian Research Council Discovery Project DP170102794. 
A preliminary version of this work appeared in the proceedings of APPROX/RANDOM2020.}}

\author{
	Catherine Greenhill\thanks{UNSW Sydney, Australia, \texttt{c.greenhill@unsw.edu.au}} \and Bernard Mans\thanks{Macquarie University, Sydney, Australia, \texttt{bernard.mans@mq.edu.au}} \and Ali Pourmiri\thanks{UNSW, Sydney, Australia, \texttt{a.pourmiri@unsw.edu.au}}}

\theoremstyle{plain}
\newtheorem{theorem}{Theorem}
\newtheorem{lemma}[theorem]{Lemma}

\newtheorem{proposition}[theorem]{Proposition}

\theoremstyle{definition}
\newtheorem{definition}[theorem]{Definition}

\theoremstyle{remark}
\newtheorem{remark}[theorem]{Remark}

\renewcommand{\leq}{\leqslant}
\renewcommand{\geq}{\geqslant}
\renewcommand{\le}{\leqslant}
\renewcommand{\ge}{\geqslant}

\newcommand{\Oh}{\mathcal{O}}
\newcommand{\E}{\mathsf{E}}

\newcommand{\e}{\mathrm{e}}
\newcommand{\I}{\mathbb{I}} %\newcommand{\I}{\mathsf{I}}

\newcommand{\C}{\mathcal{C}}

\newcommand{\A}{\mathcal{A}}

\newcommand{\B}{\mathcal{B}}
\newcommand{\Ec}{\mathcal{E}}
\newcommand{\F}{\mathcal{F}}
\renewcommand{\H}{\mathcal{H}}

\newcommand{\emp}{\textsc{{Empty}}}
\newcommand{\col}{\text{col}}
\renewcommand{\Pr}[1]{\ensuremath{\operatorname{\mathbf{Pr}}\left[#1\right]}}

\newcommand{\Ex}[1]{\ensuremath{\operatorname{\mathbf{E}}\left[#1\right]}}
\newcommand{\Expi}[2]{\ensuremath{\operatorname{\mathbf{E}}_{#2}\left[#1\right]}}
\newcommand{\Var}[1]{\ensuremath{\operatorname{\mathbf{Var}}\left[#1\right]}}

\newcommand{\viz}[1]{\ensuremath{\operatorname{\mathtt{vis}}(#1)}}

\newcommand{\unif}{s}  % uniformity = size of an edge. Used to be $\unif$.
\begin{document}
	
	\maketitle
	
	\begin{abstract}
		  We consider a variation of balls-into-bins which randomly allocates $m$ balls into $n$ bins. Following  Godfrey's model (SODA, 2008), 
		  we assume that  each ball $t$, $1\le t\le m$, comes with a hypergraph $\mathcal{H}^{(t)}=\{B_1,B_2,\ldots,B_{s_t}\}$, and each edge  $B\in\mathcal{H}^{(t)}$ contains at least a logarithmic number of bins.  
		   Given  $d\ge 2$, our $d$-choice algorithm  chooses an edge $B\in \mathcal{H}^{(t)}$, uniformly at random, and  then chooses a set  $D$ of $d$ random bins from the selected edge $B$. The ball is allocated to a least-loaded bin from $D$, with ties are broken randomly.
		   We prove that if the hypergraphs $\mathcal{H}^{(1)},\ldots, \mathcal{H}^{(m)}$ satisfy a \emph{balancedness} condition and have low \emph{pair visibility}, then  after allocating $m=\Theta(n)$ balls, the maximum number of balls at any bin, called the \emph{maximum load}, is at most $\log_d\log n+\Oh(1)$, with high probability.   
		     The balancedness condition enforces  that bins appear almost uniformly within the hyperedges of $\H^{(t)}$, $1\le t\le m$, while the pair visibility condition measures how frequently  a pair of bins is chosen during the allocation of balls.
		   Moreover, we establish a lower bound for the maximum load attained by the balanced allocation for a sequence of  hypergraphs in terms of pair visibility, showing the relevance  of the visibility parameter to the maximum load.
		  
		     In Godfrey's model,  each ball is forced to probe all bins in a randomly selected hyperedge and the ball is then allocated in a least-loaded bin. 
		 Godfrey showed that if each $\H^{(t)}$, $1\le t\le m$, is balanced and $m=\Oh(n)$, then the  maximum load is at most one, with high probability.  
		     However, we apply the power of $d$ choices paradigm, and only query the load information of $d$ random bins per ball, while achieving very slow growth in the maximum load.

		In summary, our contribution is twofold: (1) We introduce the notion of pair visibility and  study popular $d$-choice allocation algorithms with related choices including graphs and hypergraphs models. (2) In order to show the result, we generalize the witness tree technique for a dynamically changing structure, which is of independent interest.

	\end{abstract}
	
	%\newpage % CSG: Front page isn't counted within 12 page limit, so make it separate

\newpage
	
	\section{Introduction}
	The standard balls-into-bins model is a process that randomly allocates $m$ sequential balls into $n$ bins,
	where each ball chooses a set $D$ of $d$ bins, independently and uniformly at random, 
	then the ball is allocated to
	a least-loaded bin from $D$ (with ties broken randomly). 
	It is well known that when $m=n$ and $d=1$, 
	at the end of process the maximum number of balls at any bin, the \emph {maximum load}, is $(1+o(1))\frac{\log n}{\log\log n}$, with high probability. Surprisingly, Azar et al.~\cite{ABKU99} showed that for this $d$-choice process with $d\ge 2$, provided ties are broken randomly, the maximum load exponentially decreases to ${\log_d\log n}+\Oh(1)$. 
	This phenomenon is known as the \emph{power of $d$ choices}.
	The multiple-choice paradigm has been successfully applied in a wide range of problems from
	nearby server selection, and load-balanced file placement in the distributed hash table, 
	to the performance analysis of dictionary data structures (e.g., see  \cite{Wieder17}).  
	In the classical setting, all $\binom{n}{d}$ sets of $d$ bins are available to each ball.
	However, in many realistic scenarios such as cache networks, peer-to-peer or cloud-based systems, the balls (requested files, jobs, items,..)  have to be allocated to bins (servers, processors,...) that are close to them,
	in order to minimize the access latencies. 
	For example, Berenbrink et al.~\cite{BBFMN13} applied  results from balls-into-bins theory to cloud storage environments to show that it is possible to both balance the load nearly perfectly and to keep the data close to its origin.
	On the other hand, the lack of perfect randomness stimulates the de-randomization of the $d$-choice process,
	which also requires the study of non-uniform distributions over choices 
	{(e.g.~\cite{Aamand18,CRSW13,Chen19, Dahlgaard16})}. 
	Hence in many settings, allowing all possibilities for the set $D$ of $d$ bins is costly, and may
	not be practical.  This motivates the investigation of the effect of distributions of the set $D$
	on the maximum load. Therefore,
	Byers et al.~\cite{BCM04} studied a model where $n$ bins are uniformly at random placed on a geometric space. Then each ball, in turn, picks $d$ locations in the space. Corresponding to these $d$ locations, the ball probes the load of $d$ bins that have the minimum distance from the locations. The ball then allocates itself to one of the $d$ bins with minimum load. In this scenario, the probability that a location close to a bin is chosen depends on the 
	distribution of other bins in the space and hence there is not a uniform distribution over the potential choices. Here, the authors showed the maximum load is ${\log_d\log n}+\Oh(1)$. 	

	Later on, 
	Kenthapadi and Panigrahy~\cite{KP06} proposed a graphical balanced allocation in which bins are interconnected as a $\Delta$-regular graph and each ball picks a random edge of the graph. It is then placed in one of its endpoints with a smaller load. This allocation algorithm results in a maximum load of $\log\log n+\Oh\left(\frac{\log n}{\log (\Delta/\log^4 n)}\right)+\Oh(1)$.
	Here, one may see that the possibilities for the set $D$ (the two chosen bins) is restricted to the set of $n\Delta/2$ edges of the graph. 
	In the standard balls-into-bins model with $d=2$,  the underlying graph is a complete graph (all
	$\binom{n}{2}$ edges present). Kenthapadi and Panigrahy showed that 
 that the maximum load is $\Theta(\log\log n)$ if and only if 
	$\Delta={n}^{\Omega(1/\log\log n)}$. 
	  Following the study of balls-into-bins with related choices, 
	Godfrey~\cite{God08} utilized hypergraphs to model the structure of bins. He assumed that each ball comes with a $\beta$-balanced hypergraph and selects a random edge  that contain $\Omega(\log n)$ bins. Then, the ball is allocated to a least-loaded bin contained in the edge, with ties broken randomly.
	Godfrey showed that if the number of balls is bounded by $\Oh(n/(\beta\log \beta))$, then with high probability, the maximum load is one.
	Later, Berenbrink et al. \cite{BBFN12}
	improved  Godfrey's analysis and showed that if the number of balls is at most $\Oh(n/\beta)$, then the maximum load is one, with high probability.
	 Balanced allocation on graphs and hypergraphs has been further studied in~\cite{BBFN12,BSS013,PTW14,PJ19}. In the aforementioned works, either the underlying graph is fixed during the process or, in the hypergraph setting, the number of choices per ball, denoted by $d$, satisfies $d=\Omega(\log n)$.  
	 Los, Sauerwald and Sylvester~\cite{LSS} recently provided a general
	 framework for analysing balanced allocation algorithms which covers
	 several previously-studied processes, including those of~\cite{FG,MPS}.
	
		Peres et al.~\cite{PTW14} also considered balanced allocation on graphs where the number of balls $m$ can be much larger than $n$ (i.e., $m\gg n$) and the graph is not necessarily regular and dense. They established the upper bound $\Oh(\log n/\sigma)$ for the gap between the maximum and the minimum loaded bin after allocating $m$ balls, where $\sigma$ is the edge expansion of the graph. Bogdan et al.~\cite{BSS013} studied a model where each ball picks a random vertex and performs a local search from the vertex to find a vertex with local minimum load, where it is finally placed. They showed that when the graph is a constant degree expander, the local search guarantees a maximum load of $\Theta(\log\log n)$.
		When queries are allowed, Los and Sauerwald~\cite{LS22} recently proved an $O(\log \log n)$ gap for Two-Choice on dense expander graphs.
	
	Pourmiri~\cite{PJ19} substituted this local search by non-backtracking random walks of length $\ell=o(\log n)$ to sample the choices and then the ball is allocated to a least-loaded bin. Provided the underlying graph has sufficiently large girth and $\ell$, he showed the maximum load is a constant.  In the context of hashing (e.g., \cite{Aamand18, Dahlgaard16}), the authors apply the witness graph technique to analyze the maximum load in the balls-into-bins process where the bins are picked based on  tabulation.
	Recently, Bansal and Feldheim developed a global adaptive balanced allocations on
graphs that with high probability, ensures a gap of $O(d/k) \log^4 n \log \log n $ between the load of any two bins.
	
	Balanced allocation on hypergraphs can also be seen as an adversarial model, where the set $D$ of potential choices is proposed by an adversary (environment) whose goal is to increase the maximum load. Here we want to understand the
	conditions under which the balanced allocation on  (hyper)graphs still benefits from the effect of the power of $d$ choices.

	\subsection{Our Results}\label{sub:res}
	
	We study $d$-choice balanced allocation algorithms on different  environments, namely dynamic graph and hypergraph models.
	 In order to see the effect of  power of $d$-choice paradigm, we introduce the notion of \emph{pair visibility}. 
	For a pair $\{i,j\}$ of distinct vertices, the \emph{visibility} of $\{i, j\}$, denoted by $\viz{i,j}$, 
	is the number of rounds $t\in\{1,\ldots, n\}$ such $\{i,j\}$ is contained in an  edge of the $t$-th hypergraph, $\H^{(t)}$, 
	(a  formal definition is given below).
	When ball $t$ is placed into a bin, the \emph{height} of ball $t$ is the
	number of balls that were allocated to the bin before ball $t$. 
	We say that event $\E_n$ holds \emph{with very high probability} if 
	$\Pr{\E_n}\ge 1-n^{-c}$  for some constant $c>0$.
All logarithms are natural unless otherwise specified.

	\subsubsection*{Balanced Allocation on  Hypergraphs}
	
	Write $[n]=\{1,\ldots, n\}$  to be the set of $n$ bins.
	A hypergraph $\H=([n],\Ec)$ is $\unif$-\emph{uniform} if $|H|=\unif$ for every $H\in \Ec$. 
	{For every integer $n\geq 1$, let $s=\unif(n)$ be an integer such $2\leq s\leq n$.}
	A \emph{dynamic $\unif$-uniform hypergraph}, denoted by
	{$(\H^{(1)},\H^{(2)},\ldots, \H^{(n)})$,}
	is a sequence of $s$-uniform hypergraphs $\H^{(t)}=([n], \mathcal{E}_t)$ with vertex set $[n]$.
	The edge sets $\Ec_t$ may change with $t$. A hypergraph is \emph{regular} if every
		vertex is contained in the same number of edges.
	We use term ``dynamic'' to capture potential changes in hypergraphs exposed with each ball. 
	
	In this paper, we are interested in the following properties which  hypergraphs
	may satisfy.
	We refer to these properties as the \emph{balancedness}, \emph{visibility}, and 
	\emph{size} properties. 
	{The balancedness property is adapted from~\cite{BBFN12,God08}.}  
	
	\begin{definition} \label{def:balancedness}
		\underline{\sf Balancedness:}\ Let $H_t$ denote a randomly chosen edge from $\Ec_t$. 
		{If there exists a constant $\beta \geq 1$} such that 
		$\Pr{i\in H_t}\le \beta \unif/n$
		for every $1\le t\le n$ and each bin $i\in [n]$, then the dynamic hypergraph
		$(\H^{(1)},\ldots \H^{(n)})$ is $\beta$-balanced.
		A dynamic hypergraph is balanced if it is $\beta$-balanced for some constant $\beta\geq 1$.
		\end{definition}

\noindent Note that every regular hypergraph is 1-balanced.

\begin{definition} \label{def:visibility}
		\underline{\sf Visibility:} \
		For every pair of distinct vertices $\{i, j\}\subset [n]$, the visibility of $\{i,j\}$ is
		\[
		\viz{i,j}=\left|\{t \in \{1,2,\ldots, n\}\, \mid  \{i,j\}\subset H\in \Ec_t \}\right|.
		\]
		If there exists $\varepsilon = \varepsilon(n) \in (0,1)$ such that
		{$\viz{i, j} \leq \unif n^{1-\varepsilon}$} for all pairs $\{i,j\}\subseteq [n]$
		of distinct bins then the dynamic hypergraph $(\H^{(1)},\ldots,\H^{(n)})$ is  
		$\varepsilon$-visible.
		A dynamic hypergraph satisfies the visibility property if it is
		$\varepsilon$-visible for some constant $\varepsilon\in (0,1)$. 
		\end{definition}

\noindent A lower value of $\varepsilon$ means there exists a pair of nodes that  appear more frequently within the edges of $\H^{(t)}$, $t=1,\ldots,n$. 	
	 
	 \begin{definition} \label{def:size}
		\underline{\sf Size:} \
		If $\unif = \Omega(\log n)$ and 
		there exists a positive constant $c_0\geq 1$ such that  
		{$|\Ec_t| \leq n^{c_0}$} for every $t\geq 1$, 
		then
		the dynamic hypergraph $(\H^{(1)},\ldots, \H^{(n)})$ satisfies the $c_0$-size property.
		A dynamic hypergraph satisfies the size property if it satisfies the $c_0$-size
		property for some constant $c_0\geq 1$.
		\end{definition}
	
\noindent Next we define the balanced allocation process on hypergraphs.
	
	\begin{definition}[Balanced Allocation on Hypergraphs]
		Suppose that  $(\H^{(1)},\ldots, \H^{(n)})$ is an $s$-uniform dynamic hypergraph and let $d=d(n)$ be an integer-valued function such that $2\leq d\leq s$. 
		The balanced allocation algorithm on  $(\H^{(1)},\ldots, \H^{(n)})$
		proceeds in rounds $(t=1,2,\ldots, n)$, sequentially allocating $n$ balls to $n$ bins. 
		In round $t$,  the $t$-th ball chooses an edge $H_t$ uniformly at random from $\Ec_t$, %say $H_t$. After that,  
		then it randomly chooses a set $D_t$ of $d$ bins from $H_t$ (without repetition) and allocates itself to a least-loaded  bin from $D_t$, with ties broken randomly.
	\end{definition}

\begin{theorem}\label{thm:d-choice}
Let $\varepsilon = \varepsilon(n)\in (0,1)$ satisfy $n^\varepsilon\to\infty$,
		and let $d=d(n)$ be an integer-valued function such that $2\leq d = o(\varepsilon\, \log n)$.
		Let $(\H^{(1)},\ldots, \H^{(n)})$ be a sequence of  $s$-uniform hypergraphs which satisfies
		the balancedness, $\varepsilon$-visibility and size properties.
		There exists $\Omega(n)= m\le n$ such that after 
		the balanced allocation process on $(\H^{(1)},\ldots, \H^{(n)})$ has allocated $m$ balls, 
		the maximum load is $\log_d\log n+\Oh(1/\varepsilon)$
		with very high probability.
		Moreover, for every fixed positive integer $\gamma$ with $\gamma m\leq n$, after allocating $\gamma m$ balls the maximum load is at most $\gamma(\log_d\log n+\Oh(1/\varepsilon))$, with very high probability.
	\end{theorem}
	
	\begin{remark}
We are particularly interested in the case that $\varepsilon$ is constant,
and in this case we assume that $d=o(\log n)$.
	Note that when $d=\Omega(\log n)$, a constant upper bound on the maximum load
	is obtained by~\cite{God08}.
		The size property is mainly assumed for  technical reasons.  
			For instance,  $|\Ec_t|\le \text{poly}(n)$ is not necessary.  
			Roughly speaking, balanced allocation on a dynamic hypergraph with large $|\Ec_t|$ 
			resembles the 
			standard balls-into-bins process. So it might be possible that   having more structural information about a dynamic hypergraph would enable us  to extend our result to allow an arbitrary number of edges $|\Ec_t|$.  Another possible extension of
			Theorem~\ref{thm:d-choice} would be to allow
			$\unif$ to be a function of $d$. 
	\end{remark}

\underline{Proof Technique:}\ In the literature on balanced allocation, 
	 three main techniques have been applied in order to analyse 
	$d$-choice processes in the lightly loaded case, namely: layered induction, differential equations and witness trees. The first two methods make heavy use of the fact
	 that the selection of bins during the allocation process follows a uniform probability distribution over bins, which may not be the case in our model. Our main result  (i.e., Theorem~\ref{thm:d-choice}) is based on the witness tree technique, which has been applied in 
 \cite{Aamand18,Dahlgaard16,KP06,Mitzenmacher00thepower,V03}.
 The standard $2$-choice process simulates a random graph process, where any 
$2$-element subset of bins chosen by a ball can be viewed as a random edge chosen from a complete graph. Kenthapadi and Panigrahi~\cite{KP06} replaced the underlying complete graph by a dense graph, say $G$, and each ball chooses a random edge from the graph. In this model, after allocating all $n$ balls, the union of the chosen random edges builds a random subgraph of $G$. They showed that, with very high probability,  there does not exist a connected and random subgraph of size $\Omega(\log n)$ whose nodes (bins) each contain a constant  number of balls.   
  A rooted spanning tree contained in the  subgraph  is the  witness structure. Now, if the maximum load is higher  than a certain threshold,  then a deterministic construction yields the witness graph, which is a rooted tree. Thus, the maximum load is bounded from above by the threshold. In this work, although we follow the same steps as~\cite{KP06}, the witness structure is not as straightforward since the underlying structure is a sequence of  hypergraphs, $(\H^{(1)},\ldots, \H^{(n)})$, and each ball chooses $d\ge 2$ bins.
Here, the witness structure is a $d$-uniform hypergraph, say $H=(V, E)$, where (1) each node (bin) in $V$ contains a constant number of balls, (2) $|V|=\Omega(\log n)$  and (3)  $H$ has an expansion
property which means  there is an ordering of the hyperedges of $H$ so that 
{with respect to this ordering, all but a constant number of hyperedges} %but a constant number of them in this ordering 
only shares one bin with the union of the
 preceding hyperedges.
 Applying the visibility condition we conclude that, with high probability, 
there does not exist a structure satisfying these three properties. % does not exist. 
Assuming a maximum load higher than a certain threshold, we recursively build the witness structure and the proof follows.

\bigskip

The following theorem presents a lower bound for the maximum load attained by the balanced allocation on some dynamic hypergraphs in terms of $\varepsilon$-visibility.
The proof is given in Section~\ref{sec:lower-bound}.

	\begin{theorem}\label{thm:lower-bound}
Let $\varepsilon = \varepsilon(n)\in (0,1)$ satisfy $n^{\varepsilon}\to\infty$, and define 
$s=s(n)=n^{\varepsilon}$. Let $d$ be a fixed integer which satisfies $2\leq d\leq s$. 
		Then there exists a sequence of $s$-uniform hypergraphs, say $(\H^{(1)},\ldots, \H^{(n)})$, which satisfies the balancedness condition and %(trivially) satisfies 
		the $\varepsilon$-visibility condition. 
		Suppose that the balanced allocation process on 
			%for every pair $\{i,j\}\subset[n]$ we have  $\viz{i,j}\le sn^{1-\varepsilon}$. Fix a constant $2\le d\le s$ and suppose that the balanced allocation process on 
			$(\H^{(1)},\ldots, \H^{(n)})$ has allocated $n$ balls. Then 
			the maximum load is at least 
			$\min\{\Omega(1/\varepsilon),\, \Omega(\log n /\log\log n)\}$
			with probability $1-n^{-\omega(1)}$.
	\end{theorem}
	
	\subsubsection*{Balanced Allocation on Dynamic Graphs}
	
	A dynamic graph is a special case of a dynamic hypergraph, where $s=s(n)=2$ for all $n$.
	Write $(G^{(1)},\ldots, G^{(n)})$ to denote a dynamic graph,
	where $G^{(t)}=([n],E_t)$ for $t=1,2,\ldots, n$.
	Theorem~\ref{thm:d-choice} does not cover the case of graphs ($s=2$),  due to the
	size property.
	We will prove a result on balanced allocation for regular dynamic graphs.
	
	\begin{definition}[Balanced Allocation on Dynamic Graphs]
		Suppose that $(G^{(1)}, \ldots, G^{(n)})$ is a regular dynamic graph  on
		vertex set $[n]$.
		The balanced allocation algorithm on $(G^{(1)},\ldots, G^{(n)})$ 
		proceeds in rounds $(t=1,\ldots, n)$. In each round $t$, the $t$-th ball chooses an edge of
		$G^{(t)}$ uniformly at random, and the ball is then placed in one of the bins incident to the edge with a lesser load, with ties broken randomly.
	\end{definition}
	
	Say that the dynamic graph is \emph{regular} if $G^{(t)}$ is $\Delta_t$-regular
	for some positive integer $\Delta_t$ and all $t=1,2,\ldots, n$.
	For every pair of distinct bins $\{i,j\}\subset [n]$, we will assume that
	the visibility $\viz{i,j}$ satisfies
	\[
	\viz{i, j}=|\{t\in \{1,2,\ldots, n\} \mid  \{i,j\}\in E_t\}|\le 2n^{1-\varepsilon}
	\]
	for some $\varepsilon = \varepsilon(n) \in(0, 1)$. 
	This matches the definition of $\varepsilon$-visibility from
	Definition~\ref{def:visibility} for the graph case.

	\begin{theorem}\label{thm:s2c}
		Let $(G^{(1)},\ldots, G^{(n)})$ be a sequence of  regular  graphs which satisfies
		the $\varepsilon$-visibility condition, for some 
		$\varepsilon = \varepsilon(n)\in (0,1)$ such that $n^\varepsilon\to\infty$.
		Suppose that the balanced allocation process on 
		$(G^{(1)}, \ldots, G^{(n)})$ has allocated $n$ balls.  
		Then the maximum load is at most $\log_2\log n+\Oh(1/\varepsilon)$,
		with very high probability. 
	\end{theorem}
	
	The proof, which can be found in Section~\ref{sec:graph}, is again based on the witness tree technique.  We remark that Theorem~\ref{thm:s2c} can be extended to the case where the 
	dynamic graph is \emph{almost regular}, meaning that the ratio of the minimum and maximum 
	degree of $G^{(t)}$ is bounded above by an absolute constant for $t=1,\ldots, n$.

	\subsection{Dynamic Graphs and Hypergraphs with Low Pair Visibility}
	
	In order to show the ubiquity of the visibility condition, we will describe some dynamic graphs with low pair visibility. One can easily construct a dynamic hypergraph from a dynamic graph by considering the $r$-neighborhood of each vertex of the $t$-th graph as a hyperedge in the $t$-th hypergraph, for $t=1,\ldots, n$.
	
	\medskip
	
	\begin{itemize}
		\item \underline{\sf Dynamic Cycle.}
			For $t=1,\ldots, n$ define the edge set
			\[
			E_t=\{\{i, j\}\subset \{0,\ldots,n-1\} \mid  
			j = i + \lceil t/\sqrt{n} \rceil \, (\operatorname{mod}\,  n) \,\, \text{or}\,\, 
			i =  j + \lceil t/\sqrt{n} \rceil\, (\operatorname{mod}\, n) \},
			\]
			where calculations are performed modulo $n$ (that is, in the additive group $\mathbb{Z}_n$).
			In modular addition, for every pair $\{i, j\}\subset \{0,\ldots, n-1\}$, the equation 
			$i =  j+ k\, (\operatorname{mod}\, n)$ has at most one solution $1\le k\le \sqrt{n}$ and hence 
			\[
			\viz{i, j}=|\{t\in \{1,2,\ldots, n\}\mid  \{i,j\}\in E_t\}|\le \sqrt{n}.
			\]
			Now $C^{(t)} = (\{0,1,\ldots, n-1\},E_t)$ is 2-regular, so it is either a Hamilton
			cycle or a union of two or more disjoint cycles
			(depending on whether $t$ and $n$ are coprime).
			By Theorem~\ref{thm:s2c}, w.h.p.\ the maximum load attained by the algorithm on $\{C^{(t)}, t=1,\ldots,n\}$ is at most $\log_2\log n+\Oh(1)$. 
				The analysis of the balanced allocation algorithm on $\Delta$-regular graphs given by 
				Kenthapadi and Panigrahy~\cite{KP06} showed that the balanced allocation process
				on arbitrary $\Delta$-regular graphs has maximum load $\Theta(\log \log n)$ only when
	$\Delta = n^{\Omega(1/\log\log n)}$.  By contrast, here each $C^{(t)}$ has degree at most $2$,
			but 
			the visibility condition keeps the maximum load as low as the standard 
			two-choice process.
			
			\smallskip
			
			\begin{remark}
				By Theorem~\ref{thm:s2c}, w.h.p.\ the balanced allocation process on  the dynamic cycle  achieves  the maximum load at most  $\log_2\log n+\Oh(1)$. 
	Since $|E_t|=n$ for $t=1,\ldots, n$, each ball requires $\log_2 n$ 
		random bits. However, in the standard power-of-two-choices process, each ball chooses two independent and random bins, which requires $2\log_2 n$ random bits.
				Therefore, the dynamic cycle can be  used to reduce (by half) the number of
				random bits required in the standard two-choice process. 
			\end{remark}
		
		\medskip
		\item\underline{\sf Dynamic Modular Hypergraph.}
		Suppose that $n$ is a prime number and  fix $s=s(n)$  such that $\log n\leq s\leq n^{1/5}$. (Here $n$ is large enough so that this range is non-empty.)
		For $t=1,\ldots,n$, let $k_t=\lceil \sqrt{n} \rceil + \lceil \frac{t}{n^{3/4}} \rceil$ and 
		for each $\alpha\in \mathbb{Z}_n$
		define 
		\[
		H_t(\alpha)=\{ \, \alpha + j k_t\, (\operatorname{mod}\, n) \mid  j=0,1,\ldots,s-1 \, \}.
		\]
		Then $H_t(\alpha)$ is a subset of $\mathbb{Z}_n$ of size $s$, as $n$ is prime.
		Now for each $t=1,\ldots,n$ we define the dynamic $s$-uniform hypergraph $\H^{(t)}=(\mathbb{Z}_n, \mathcal{E}_t)$, 
		where 
		$\mathcal{E}_t=\{H_t(\alpha) \mid \alpha\in \mathbb{Z}_n\}$.  Then $\H^{(t)}$ is $s$-regular, and hence
		1-balanced, and it satisfies the 1-size property as $|\mathcal{E}_t|= n$.
		Suppose that $\{\beta_1, \beta_2\}\subset H_t(\alpha)$ for some $\alpha\in \mathbb{Z}_n$,
		{with $\beta_1\neq \beta_2$}. 
		Then there exists $j_1,j_2\in \{0,\ldots, s-1\}$ such that 
		$\beta_1 = \alpha  + j_1 k_t\, (\operatorname{mod}\, n)$ and 
		$\beta_2 = \alpha + j _2 k_t\, (\operatorname{mod}\, n)$. 
		Thus, $\beta_2 - \beta_1 =  (j_2-j_1) k_t\, (\operatorname{mod}\, n)$. Note that $j_1, j_2$ must
			be distinct as $\beta_1,\beta_2$ are distinct.
		Next suppose that $k_{t_1}\neq k_{t_2}$ for some $t_1,t_2\in\{1,\ldots, n\}$, and take any $j_1,j_2\in \{1,\ldots, s-1\}$.  
			{By definition of $k_t$} and working in $\mathbb{Z}$, we see that
			\[ 1\leq | j_2 k_{t_2} - j_1 k_{t_1} | \leq (s-1)\big( \lceil\sqrt{n}\rceil + \lceil n^{1/4} \rceil\big) < n,
			\]
		and it follows that 
		\begin{equation}
		\label{no-mod}
		j_1k_{t_1}\neq j_2k_{t_2}\, (\operatorname{mod}\, n).
		\end{equation}
		{Finally, suppose that some distinct $\beta_1,\beta_2$ satisfy
			$\{\beta_1, \beta_2\}\subset H_{t_1}(\alpha)\cap H_{t_2}(\alpha)$ where $k_{t_1}\neq k_{t_2}$.
			Then $\beta_2 - \beta_1 = j k_{t_1}\, (\operatorname{mod}\, n)$ for some $j_1\in\{1,\ldots, s-1\}$, 
			and $\beta_2 - \beta_1 = j_2 k_{t_2}\, (\operatorname{mod}\, n)$ for some $j_2\in \{ 1,\ldots, s-1\}$,
			but this contradicts (\ref{no-mod}).
		}
		Therefore, by definition of $k_t$, for every $\{\beta_1, \beta_2\}\subset\mathbb{Z}_n$, we have
		\[
		\viz{\beta_1, \beta_2}=|\{t\in \{1,2,\ldots, n\}\mid  \{\beta_1, \beta_2 \}\subset H_t(\alpha) ~\text{for some}~ \alpha\in \mathbb{Z}_n\}| = \Oh(n^{3/4}).
		\]
		
		\item \underline{\sf Stationary Geometric Mobile Network.} 
		Fix a constant integer $R\geq 2$.
		Consider an $R$-dimensional torus $\Gamma(n, R)$, which is a graph whose
		vertex set is the Cartesian product of $\mathbb{Z}_\ell^R={\mathbb{Z}_\ell\times\ldots\times\mathbb{Z}_\ell}$, where $\ell=n^{1/R}\in\mathbb{Z}$, and two vertices 
		$(x_1,\ldots,x_R)$ and $(y_1,\ldots,y_R)$ are connected if for some $j\in \{1,2\ldots,R\}$ $x_j=y_j\pm 1$ mod $n$ and for all $i\neq j$ we have $x_i=y_i$. 
		Let $\pi$ be the stationary distribution of
		the following random walk on $\Gamma(n, R)$: at each step, the walker stays at the
		current vertex with probability $p$, and otherwise chooses a neighbour randomly
		and moves to that neighbour.  
		The transition probability from vertex $u$ to a neighbouring vertex $w$ is $(1-p)/(2R)$, where $2R$ is the degree of vertex $u$ in $\Gamma(n,R)$.
		Now place $n$ agents on vertices of $\Gamma(n,R)$ independently, each according to the distribution $\pi$.  
		At each time step, each agent independently performs a step of the random walk
		described above (For random walks on a torus we refer the interested reader to \cite{LPW06}). For every pair of distinct agents $a$ and $b$, let $d_t(a,b)$ denote the Manhattan distance (in $\Gamma$) of the locations of $a$ and $b$ at time $t$. For a given $r\ge 1$, we define the \emph{communication graph process} $\{G^{(t)}_r \mid  t=0,1,\ldots\}$ over the set of agents, say $A$, so that for every $t\ge 0$, agents $a$ and $b$ are connected if and only if $d_t(a,b)\le r$. The model has been thoroughly studied when $R=2$ in the context of information spreading~\cite{CMPS11}.
		We present the following result regarding the pair visibility of the communication
		graph process, proved in Section~\ref{sec:pop}. 
		
		\smallskip
		\begin{proposition}\label{pro:gmn}
			Fix $r=r(n) = n^{o(1)}$. Also let $\{G_r^{(t)}=(A, E_t) ~|~ 1\le t\le n\}$ be the communication graph process defined on an $R$-dimensional torus $\Gamma(n, R)$, where $R\geq 2$ is a fixed integer. Then there exists constant $\varepsilon> 0$ such that for every pair of agents, say $\{a, b\}\subset A$,
			\[
			\viz{a, b}=|\{t \in \{1,2,\ldots, n\} \mid  \{a,b\}\in E_t\}|=\Oh(n^{1-\varepsilon}).
			\]
		\end{proposition} 
		
	\end{itemize}

	\section{ Balanced Allocation on  Hypergraphs }\label{sec:hyp}
	In this section we establish an upper bound for  the maximum load attained by the balanced allocation on hypergraphs (i.e., Theorem~\ref{thm:d-choice}). 
	In order to analyze the process  let us first define a \emph{conflict graph}. 
	We write $D_t$ for the set of $d$ bins chosen by the $t$-th ball,
	and sometimes refer to $D_t$  as the $d$-\emph{choice} of the $t$-th ball.
	We will slightly abuse the notation and write $D_u\cap D_t$,\,  $D_u\cup D_t$ to denote the 
	set of common bins, and the union of bins,  chosen by balls $u$ and $t$, respectively. 
	
	\begin{definition}[Conflict Graph]
		For $m=1,\ldots, n$, the conflict graph $\C_m$ is a simple graph with vertex set 
		$\{D_1,D_2,\ldots,D_m\}$. 
		Vertices $D_u$ and $D_t$ are connected by an edge in $\C_m$
		if and only if $D_u\cap D_t\neq \emptyset$  (that is, the $d$-choices of
		the $t$-th ball and the $u$-th ball contain a common bin).
	\end{definition}
	
	We say a subgraph of $\C_m$ with vertex set $\{D_{t_1},\ldots, D_{t_k}\}$ is $c$-\emph{loaded} if every bin in $D_{t_1}\cup D_{t_2}\cup \cdots \cup D_{t_k}$ has at least $c$ balls. 
	
	Our analysis will involve
	a useful combinatorial object, called an {\it ordered tree}. 
	An ordered tree is a rooted tree, together with a specified ordering of the children of every vertex.
	Recall that $\frac{1}{k+1}\, \binom{2k}{k}$ is the $k$-th Catalan number, which counts numerous combinatorial objects, including the number of ways to form $k$  balanced parentheses.
	It is well known~\cite{OEIS} that ordered trees with $k-1$ edges are 
	counted by the $(k-1)$-th Catalan number, leading easily to the following proposition.
	
	\begin{proposition}\label{pro:ordered}
		The number of $k$-vertex ordered trees is $\frac{1}{k}\,\binom{2k-2}{k-1}\leq 4^{k-1}$.
	\end{proposition}
	More information regarding  the enumeration of trees can be found in~\cite[Section~2.3]{Knuth}. 
	
	The following blue-red coloring will be very helpful in our analysis.
	
	\begin{definition}[Blue-red coloring]\label{def:br} Given $m\in \{1,2,\ldots, n\}$,
		suppose that $T\subset \C_m$ is a  rooted and ordered $k$-vertex tree contained in $\C_m$.
		Let the vertex set of $T$ be $\{D_{t_1}\ldots,D_{t_k}\}$, where $D_{t_1}$ is the root. 
		Perform depth-first search starting from the root, respecting the specified order of
		each vertex. For $i=1,\ldots, k$, let $D(i)\in\{D_{t_1}\ldots,D_{t_k}\}$ be the vertex
			which is the $i$-th visited vertex in the depth-first search.  Then $D(1)=D_{t_1}$ is the root.
		We now define a  blue-red coloring  $\col:\{D(2),\ldots, D(k)\}\rightarrow \{\text{blue, red}\}$ as follows. For $i= 2,\ldots, k$,
		\[
		\col(D(i))= \begin{cases}
		\text{blue} &  \text{ if }\,\, |\, (\cup_{j=1}^{i-1}D(j)) \cap D(i)|= 1,\\
		\text{red} &  \text{ if }\,\, |\, (\cup_{j=1}^{i-1}D(j)) \cap D(i)|\ge 2.
		\end{cases}
		\]
	\end{definition}
	
	The following key lemma presents a upper bound for the probability that a certain tree 
	can be found as a subgraph of $\C_m$.
	
	\begin{lemma}[Key Lemma]\label{lem:col} 
			Let $\beta,\, c_0\geq 1$ be constants and let 
				$\varepsilon = \varepsilon(n)\in (0,1)$ such
				that $n^\varepsilon\to\infty$. 
				Let $d=d(n)$ be an integer-valued function such that
				$2\leq d = o(\varepsilon\, \log n)$.
		Let $(\H^{(1)},\ldots, \H^{(n)})$ be a dynamic $s$-uniform hypergraph which satisfies the
			$\beta$-balanced,
			$\varepsilon$-visibility and $c_0$-size properties.
Suppose that $c \geq 44\beta \e^2$ is a sufficiently large constant
and $k=C\log n$ for some constant
		$C\geq 1$. There exists $\Omega(n)= m\le n$ such that the probability that $\C_{m}$ contains a $c$-loaded  $k$-vertex tree with $r$ red vertices in its blue-red colouring is at most
		\[
		n^{c_0+3}\, \exp\{4k\log({2\beta d})-r\varepsilon \log(n)/2-c(d-1)(k-r-1)\}.
		\]
		%where $r$ is the number of red vertices in the  blue-red coloring of the tree. 
		Moreover, with probability $1-n^{-\omega(1)}$ no such tree exists with $\omega(1/\varepsilon)$ red vertices. 
	\end{lemma}

	The proof, presented in Section~\ref{sub:exist}, involves an extension of the
	witness tree technique.  This method might be of independent interest in the study of random hypergraphs. 
	
	We now explain how to recursively build a witness graph if there exists a bin whose load is higher than a certain threshold.
	The {\it minimum load} of $D_t$ is the number of balls in the least-loaded bin in $D_t$
	(the set of $d$ choices of $D_t$).
	Clearly, if ball $t$ is placed at height $h$ then $D_t$ has minimum load at least $h$.

	\paragraph*{Construction of the Witness Graph.}{Suppose that there exists a bin with load $c+\ell+1$. Let $R$ be the $d$-choice corresponding to the ball at height $c+\ell$ in this bin. 
	When this ball arrived, the minimum load of $R$ was $c+\ell-1$. 
		We start building the witness tree in $\C_m$ whose root is $R$. 
		For every bin $i\in R$, consider the $\ell$ balls in bin $i$ at height $c+\ell- j$, 
		for $j=1,\ldots, \ell$, and let $D^{i}_{t_j}$ be the $d$-choice corresponding to the
		ball in bin $i$ with height $c+\ell-j$.  
		(Here $t_j$ is the time of arrival of the ball in bin $i$ at height $c+\ell-j$.)
		These $\ell$ balls exist as the minimum load
		of $R$ is $c+\ell$.
		We refer to the set $\{D^i_{t_j} \mid i\in R,\,\,  j=1\ldots \ell\}$ 
		as the set of \emph{children} of $R$, where the minimum load of $D^i_{t_j}$ is $c+\ell-j-1$.  All children of $R$ are connected to $R$ in  $\C_{m}$.  Order the children of
		$R$ arbitrarily, then blue-red colour the first level of the tree (the children of $R$).
		Recall that a vertex is colored by blue  if it only shares  one bin with its predecessors in the ordering. So a blue $d$-choice contains $d-1$ bins that have not appeared in
		previous $d$-choices (with respect to depth-first search, respecting the fixed ordering).
		We call these $d-1$ bins {\it fresh}.
		
		Next, consider each blue leaf of the tree (if any), and recover the $d$-choices corresponding to  balls that are placed in fresh bins with height at least $c$. Then, blue-red color the children of those $d$-choices, with respect to an arbitrary ordering. This recursion will continue until either there are no balls remaining with height at least $c$, or there are no blue leaves. 
		For $j=0,\ldots, \ell$,
		let $f(\ell-j)$ denote the number of $d$-choices that the recursive construction gives, assuming that the $d$-choice for the root has minimum load  $c+\ell-j-1$ and that all balls in the tree are blue. In this case, the recursive construction continues until no ball remains with height at least $c$. 
		We claim that 
		\[
		f(\ell)\ge (d-1)(f(\ell-1)+f(\ell-2)+\cdots + f(0)+1),
		\] 
		where $f(0)=1$. 
		To see this, suppose that bin $B$ contains $c+\ell$ balls, and consider
		the $\ell+1$ balls in this bin at height $c$, $c+1$, \ldots, $c+\ell$.  The ball at height $c+\ell-j$ corresponds to a $d$-choice which contains $d-1$ additional bins whose minimum load was at least $c+\ell+j-1$ at the time when ball $c+\ell-j$ arrived. Hence the ball at height $c+\ell-j$ in $B$ implies the existence of $(d-1)f(\ell-j-1)$ bins in the recursive construction, for $j=0,1,\ldots, \ell-1$, in addition to the $d-1$ fresh bins corresponding to $B$. These bins are all distinct, by the assumption that all vertices are blue. This implies the claimed recursive lower bound on $f(\ell)$. Then it follows by strong induction that
		{$f(\ell)\ge 2(d-1)d^{\ell-1}\geq d^\ell$}.
	}

	\begin{proof}[Proof of Theorem~\ref{thm:d-choice}]
		Let $(\H^{(1)},\ldots, \H^{(n)})$ be a dynamic hypergraph which satisfies
		the $\beta$-balanced, $\varepsilon$-visibility and $c_0$-size properties.
		By Lemma~\ref{lem:col}, there exists $\Omega(n)=m\le n$ such that
		the following holds with probability $1-n^{-\omega(1)}$: 
		after $m$ balls have been allocated by the balanced allocation process,
		if $T\subseteq \C_m$ is a $c$-loaded tree with $k$ vertices and $T$ is blue-red coloured
		according to some arbitrary ordering of the children of each vertex, then the number
		$r$ of red vertices satisfies $r=\Oh(1/\varepsilon)$.
		Let $c_2 = r/d = \Oh(1/\varepsilon)$.

		Now suppose that after allocating $m$ balls, there is a ball at height  $c_1+c_2 + \ell+1$, where $c_1$ is some sufficiently large but fixed positive integer. This implies that there is a $d$-choice, denoted by $R$, whose minimum load is at least $\ell+c_1+c_2+1$. Let us consider all balls placed in the bins contained in $R$
		with height at least $c_1 + \ell+1$. Recover the  corresponding $d$-choices for 
		these balls, say $D_1,D_2,\ldots, D_w$, then colour them blue-red with respect to the root $R$ and an arbitrary ordering of the children of each vertex. 
		Note that $w\ge c_2\cdot d$, as each bin in $R$ contains at least
		$c_2$ balls at height at least $c_1+\ell+1$.
		Suppose that among these $w$ there are $b\ge 1$ blue vertices and $w-b$ red vertices. We now consider  every blue vertex $D_t\in \{D_1,D_2,\ldots, D_w \}$ as a root and start  the recursive construction of the witness graph. Assuming that the number of red vertices is strictly less than $c_2\cdot d< w$,  it follows that at least one recursive construction (with root $D_i$) does not produce any red vertex. Moreover, the recursion from $D_i$ gives a $c_1$-loaded tree with at least $k=d^\ell$ vertices. We take $\ell = \log_d\log n$, so that $k=\log n$.
		Another application of Lemma~\ref{lem:col} implies that a $c_1$-loaded $k$-vertex tree 
		with no red vertices exists with probability at most 
			\begin{align*}  
			n^{c_0+3}\, \exp\{4k\log(2\beta d)-c_1(d-1)(k-1)\} &\leq 
			\exp\left\{ \big(c_0 + 4 + 4\log(2\beta d) - c_1(d-1)\big)\log n \right\}\\
			&\leq \exp\left\{ \big(c_0  + 4 + 4\log(4\beta) - c_1\big)\log n\right\},
			\end{align*}
			using the fact that $2\leq d = o(\log n)$ and $k=\log n$.
			Setting $c_1$ to be a large enough positive constant,
		we conclude that with very high probability the maximum load is at most 
		\[
		\ell + c_1 + c_2 = {\log_d\log n+\Oh(1)+c_2} = \log_d\log n+\Oh(1/\varepsilon),
		\]  
		since $c_2=\Oh(1/\varepsilon)$.
		This proves the first statement of Theorem~\ref{thm:d-choice}.
		%{The proof of the second statement is presented in Appendix \ref{app:missmain}.}

%%% The rest of the proof used to be Appendix D
	In order to prove the second statement of Theorem~\ref{thm:d-choice} we show the sub-additivity of the balanced allocation algorithm. 
	We want to prove that for every constant integer $\gamma \ge 1$ with $\gamma m \leq n$, after allocating $\gamma m$ balls, the maximum load  is at most $\gamma(\log_d\log n+\Oh(1/\varepsilon))$, with high probability. 
	First assume that $2m\leq n $ and suppose that the algorithm has allocated $m$ balls to $\H^{(t)}, t=1,\ldots,m$ and let $\ell^*\le \log_d\log n +\Oh(1/\varepsilon)$ denote its maximum load. 
	We now consider two independent  balanced allocation algorithms, say $\A$ and $\A_0$, on two dynamic hypergraphs starting from step $m$. 
	These dynamic hypergraphs are $(\H^{(m)},\ldots, \H^{(n)})$ and $(\H_0^{(m)},\ldots, \H_0^{(n)})$, where $\H_0^{(t)}$ is an identical copy of $\H^{(t)}$ for $t=m,\ldots, n$. 
	Moreover, we assume that in round $m$, all bins contained in $\H^{(m)}_0$ have 
	{exactly} $\ell^*$ balls.
	Let us couple algorithm $\A$ on $\H^{(t)}$ and algorithm $\A_0$ on $\H_0^{(t)}$. 
	Write $V=[n]$ for the set of $n$ bins.
	To do so, the coupled process allocates a pair of balls to bins as follows:
	for $t=m+1,\ldots, n$,  the coupling chooses
	a one-to-one labeling function \mbox{$\sigma_t: V\rightarrow \{1, 2,\ldots,n\}$} uniformly at random, where $V$ is the ground set of both hypergraphs (i.e, set of $n$ bins) and $\{1, 2,\ldots,n\}$ is a set of labels.
	Next, the coupling chooses $D_{t}$ randomly from $\H^{(t)}$. 
	{Let $D'_t$ denote the  same set of $d$ bins as $D_t$ in $\H_0^{(t)}$.}
	Algorithm $\A$ allocates ball $t+1$ to a least-loaded vertex of $D_{t}$,
	and algorithm $\A_0$ allocates ball $t+1$ to a least-loaded vertex of $D'_t$,
	with both algorithms breaking ties in favour of the vertex $v$ with the
	smallest load and minimum label $\sigma_t(v)$.
	Note that algorithm $\A$ is a faithful copy of the balanced allocation process on 
	$(\H^{(m)},\ldots, \H^{(n)})$, and algorithm $\A_0$ is a faithful copy of the 
	balanced allocation process on $(\H_0^{(m)},\ldots, \H_0^{(n)})$,
	respectively.  (This follows as $\sigma_t$ is chosen uniformly at random.)
	Let $X_i^{t}$ and $Y_i^{t}$, $m+1\le t\le 2m$, denote the load of bin $i$ in $\H^{(t)}$ and $\H^{(t)}_0$, respectively.
	We prove by induction that for every integer $m\le t\le 2m$ and $i\in V$ we have 
	\begin{align}\label{domin1}
	X_i^{t}\le Y_i^{t}.
	\end{align}
	The inequality holds when $t=m$ by the assumption that $Y_i^m=\ell^*$ for every $i\in V$.
	Let us  assume that for every $t'$ with $t'\le t< 2m$, Inequality (\ref{domin1})  holds,  then we will show it for $t+1$. Let  $i\in D_{t+1}$ and $j\in D'_{t+1}$ denote the  vertices (bins) that receive a ball in step $t+1$. We now consider two cases:
	\begin{itemize} 
		\item {\bf Case 1:} $X_i^{t}<Y_i^{t}$.  Since algorithm $\A$ allocated ball $t+1$ to bin $i$, it follows that  
		\[
		X_i^{t}+1=X_i^{t+1}\le Y_i^{t}\le Y_i^{t+1}.
		\]
		So, Inequality (\ref{domin1}) holds for $t+1$ and every  bin $i\in V$.
		\item {\bf Case 2:} $X_i^{t}=Y_i^{t}$.
		Since $D'_{t+1}$ is a copy of $D_{t+1}$, we have $j\in D_{t+1}$ and $i\in D'_{t+1}$. We know that no vertex (bin) in $D_{t+1}$ has smaller load than $i$, and 
		no vertex (bin) in $D'_{t+1}$ has smaller load than $j$.  Hence
		\[
		X_{i}^{t} \le X_{j}^{ t} \leq Y^{t}_{j}\le Y^{t}_{i}, 
		\]
		{where the middle inequality follows from the inductive hypothesis (\ref{domin1}) for bin $j$}.
		So by assumption of this case  we obtain $
		X_{i}^{t}=X_{j}^{ t}= Y^{t}_{j}=Y^{t}_{i}$.
		If  $i\neq j$ and  $\sigma_{t+1}(j)< \sigma_{t+1}(i)$, then it contradicts  the fact that ball $t+1$ is allocated to bin $i$ by algorithm $\A$. Similarly, if $\sigma_{t+1}(j)> \sigma_{t+1}(i)$, then it contradicts the fact that 
		{algorithm $\A_0$} allocated ball $t$ to bin $j$. Therefore $i=j$ and hence 
		\[
		X_i^{t+1} = X_{i}^{t}+1=Y^{t}_i+1=Y_i^{t+1}.
		\]
	\end{itemize}
	Thus, in both cases,  Inequality (\ref{domin1}) holds for every $t\ge 0$.
	By applying the first part of the theorem, with very high probability,
	using algorithm $\A_0$ to allocate $m$ 
	balls to the dynamic hypergraph $(\H_0^{(m)},\ldots, \H_0^{(2m)})$ results in maximum load 
	\[
	\ell^*+\log_d\log n+\Oh(1/\varepsilon)\le 2(\log_d\log n+\Oh(1/\varepsilon))
	\] 	
	{in $\H_0^{(2m)}$}.
	Therefore, by Inequality (\ref{domin1}), after using algorithm $\A$
	to allocate $m$ balls to the dynamic hypergraph $(\H^{(m)},\ldots, H^{(n)})$,
	with very high probability
	the maximum load in $\H^{(2m)}$ is at most  $2(\log_d\log n+\Oh(1/\varepsilon))$. 
	Applying  the union bound, we conclude that after allocating  $\gamma m$ balls, 
	where $\gamma m \leq n$,
	the maximum load  is at most $\gamma(\log_d\log n+\Oh(1/\varepsilon))$, 
	with very high probability.
This completes the proof of Theorem~\ref{thm:d-choice}.	
	\end{proof}

\section{Appearance Probability of a Certain Structure}\label{sub:exist}
	
	In this section we work towards a proof of Lemma~\ref{lem:col}. 
	First we will give a useful definition and prove some helpful results.
	The definition  was introduced in \cite[Definition~3]{PJ19}.
	
	\begin{definition}\label{def:uniform}
		Suppose that $\A$ is an allocation algorithm that sequentially allocates $n$ balls into $n$ bins according to some mechanism. 
		For a given constant $\alpha > 0$, and for $\Omega(n) = m \leq n$,
		we say  that $\A$ is $(\alpha, m)$-uniform if for every ball $1\le t\le m$ and every bin $i\in [n]$,
		\[
		\Pr{\,\text{ball $t$ is allocated to bin $i$ by $\A$} \mid \text{balls $1,2,\ldots, t-1$ have been allocated by  $\A$}\,}\le \frac{\alpha}{n}.
		\]
	\end{definition}
	{In the above definition, we condition on the allocations of balls $1,\ldots, t-1$ into bins made by~$\mathcal{A}$. }
	
\bigskip

%%%%%%%%%%%%%%%%%%%
%%%%%%%%%%%% This stuff used to be in a short appendix.

	Berenbrink et al.~\cite{BBFN12} proposed an allocation algorithm $\B$ such that 
	for $t=1,2,\ldots$, the $t$-th ball chooses an edge of $\H^{(t)}=([n], \Ec_t)$, uniformly at random, say  $H_t$.  The ball is then allocated to an empty vertex (bin) of $H_t$, with ties broken randomly.  
If $H_t$ does not contain an empty bin then the process fails.
	The next lemma follows directly from~\cite[Lemmas 4, 5]{BBFN12}. 
	
	\begin{lemma}\label{lem:balls}
Suppose that the dynamic $\unif$-uniform hypergraph $(\H^{(1)}, \ldots, \H^{(n)})$
satisfies the balancedness and size properties. There exists $m=\Theta(n)$ such that 
with probability at least $1-n^{-2}$, algorithm $\B$ successfully allocates $m$ balls and
there are at least $\unif/2$ empty vertices in $H_t$ for $t=1,\ldots, m$. 
\end{lemma}

	We now apply the above result to %prove Lemma~\ref{lem:empty}.
 show the same property holds for the balanced allocation on any dynamic hypergraph. 
	
	\begin{lemma}\label{lem:empty}
		Fix $d=d(n)$ with $2\leq d = o(\log n)$. 
		Suppose that the dynamic $\unif$-uniform hypergraph $(\H^{(1)}, \ldots, \H^{(n)})$
		satisfies the balancedness and size properties. 
		There exists $m = \Theta(n)$ with $m < n$ such that 
		with probability at least $1-n^{-2}$, 
		the edge $H_t$ chosen by the $t$-th ball contains at least
		$\unif/2$ empty vertices for $t=1,\ldots, m$.
	\end{lemma}

	\begin{proof}%[Proof of Lemma~\ref{lem:empty}]
		We apply a coupling technique between the balanced allocation process on a dynamic hypergraph and $\B$.  
		
		Let us first consider an identical copy of the set of bins, called $B$. 
		The coupled process sequentially allocates a ball to a pair of bins.
		In round $t=1,\ldots, m$, the $t$-th ball chooses an edge of $\H^{(t)}$ uniformly at random, say $H_t$. Let $H'_t$ be the corresponding set of bins, chosen from $B$. Then the first ball is allocated to a bin, say $i$, contained in $H_t$ according to the balanced allocation. If $i\in H'_t$ is empty then the second ball is allocated to bin $i\in H'_t$ as well. If $i\in H'_t$ is not empty then the second  ball is allocated to an empty bin from $H'_t$, with ties  broken randomly. If there is no empty bin in $H'_t$ then the coupling fails.
		Note that $H_t$ and $H'_t$ have the same set of bins but may have different loads. 
		Observe that the coupled process allocates balls to bins from $B$ according to $\B$. 
		Next we show that for $t=1,\ldots, m$,
		\begin{align}\label{ineq:ind}
		\emp(H_t)\ge \emp(H'_t),
		\end{align}
		where $\emp(H)$ denotes the number of empty bins contained in $H$.
		For a contradiction, assume that there is a first time $t_1$  such that 
		$\emp(H'_{t_1})>\emp(H_{t_1})$. Then there is vertex $i\in H'_{t_1}$ which is empty, while $i\in H_{t_1}$ has a ball at height zero: this is ball $t_0$, say, where $1\le t_0\le t_1$. 
		This implies that the coupled process has allocated ball $t_0$ to bin $i\in H_{t_1}$,
		but it has not allocated any ball to bin $i\in H'_{t_1}$, since $i$ was empty until round $t_1$.  This contradicts the definition of the coupled process. 
		So Inequality (\ref{ineq:ind}) holds for $t=1,\ldots, m$.    
		Applying  Lemma~\ref{lem:balls} yields that there exists $m=\Theta(n)$ such that 
		$\emp(H_t)\ge \emp(H'_t)\ge \unif/2$
		for $t=1,\ldots, m$.
	\end{proof}

Using this result we can prove that the balanced allocation process is $(\alpha,m)$-uniform on dynamic hypergraphs.

	\begin{lemma}[Uniformity Lemma]\label{lem:uni}
		Fix $d=d(n)$ with $2\leq d = o(\log n)$ and suppose that 
			for some constant $\beta \geq 1$,
			the $s$-uniform dynamic hypergraph $(\H^{(1)},\ldots, \H^{(n)})$ satisfies the
			$\beta$-balanced and size properties, with $d\leq s$.  
			Then there exists a constant $\alpha=\alpha(\beta)$
			%, which depends only on $\beta$, 
			and there exists $m=\Theta(n)$ with $m \leq n$,
			such that the balanced allocation process on $(\H^{(1)},\ldots, \H^{(n)})$
			is $(\alpha,m)$-uniform.  Specifically, we may take $\alpha = 44\beta$.
	\end{lemma}

	\begin{proof}
Fix $m=m(n)$ to equal the $m$ provided by Lemma~\ref{lem:empty}.
		For $t=1,\ldots, m$, let $D_t$ be the $d$-element subset of $H_t$ that is chosen by the $t$-th ball. 
		Define the indicator random variable $\I_t$ as follows:
		\[
		\I_t:= \begin{cases}
		1 & \text{ if $D_t$ contains at least $d/6$ empty vertices, }\\
		0 & \text{ otherwise.}\\
		\end{cases} 
		\]  
		Let us fix an arbitrary bin $i$ and then  define $A(t,i)$ to be the event that the
		$t$-th ball is allocated to vertex $i$.  (The first $t-1$ balls have already been
			allocated, as the balanced allocation process allocates one ball per step.)
		Observe that if $i\not\in D_t$ then $\Pr{A(t,i)} = 0$.  It follows that
		\begin{align*}
		\Pr{A(t,i)}&=\Pr{A(t,i)\mid i\in D_t~ \text{and}~ \I_t=1}\cdot \Pr{i\in D_t ~\text{and}~ \I_t=1}\\
		& {} \quad + \Pr{A(t,i) \mid i\in D_t ~\text{and} ~\I_t=0}\cdot\Pr{i\in D_t~ \text{and}~ \I_t=0}.
		\end{align*}
		Now there are at least $d/6$ empty vertices in $D_t$ when $\I_t=1$, so 
		\[ \Pr{A(t,i)\mid i\in D_t~ \text{and} ~ \I_t=1} \leq 6/d.\]
		It follows that
		\begin{align}
		\Pr{A(t,i)}
		 &\le (6/d)\, \Pr{i\in D_t ~\text{and}~ \I_t=1} + \Pr{i\in D_t~ \text{and}~ \I_t=0}\nonumber \\
		 &\le (6/d)\, \Pr{i\in D_t } + \Pr{i\in D_t~ \text{and}~ \I_t=0}.\label{ineq:prefirst}
		\end{align} 
		In order to have $i\in D_t$, first an edge containing $i$ must be selected, and then the chosen $d$-element subset of that edge must contain $i$. By the $\beta$-balancedness property, 
		\begin{equation}
		\label{eq:beta-balancedness}
		\Pr{i\in D_t}\le \frac{\beta \unif}{n}\cdot \frac{\binom{\unif-1}{d-1}}{\binom{\unif}{d}}\le \frac{\beta d}{n}.
		\end{equation}
		First suppose that $d\leq 6$.
		Then it follows from (\ref{ineq:prefirst}) and (\ref{eq:beta-balancedness}) that
		\[
		\Pr{A(t,i)}\le \big( 6/d + \Pr{\I_t=0\mid i\in D_t}\big)\, \Pr{i\in D_t} \leq \frac{6\beta}{n} + \frac{\beta d}{n}\, \leq \frac{12\beta}{n}.
		\]
Taking $\alpha = 12\beta$,  this completes the proof when $d\leq 6$. 

		For the remainder of the proof we assume that $d\geq 7$. 
		Let $\F$ be the event that $H_t$ contains at least $\unif/2$ empty vertices
		for all $t=1,\ldots, m$. 
		By Lemma~\ref{lem:empty}, we have $\Pr{\F}\ge 1-n^{-2}$. 	
		Then 
\begin{align}
  & \hspace*{-1cm}\Pr{i\in D_t~ \text{and}~ \I_t=0}\nonumber \\
&= \Pr{i\in D_t~ \text{and}~ \I_t=0 \mid \F}\cdot \Pr{\F} +
   \Pr{i\in D_t~ \text{and}~ \I_t=0 \mid \neg \F}\cdot \Pr{\neg \F} \nonumber \\
&\leq \Pr{i\in D_t~ \text{and}~ \I_t=0 \mid \F} + \Pr{\neg \F} \nonumber \\
&\leq 2\Pr{i\in D_t~ \text{and}~ \I_t=0 \text{~and~} \F} + n^{-2}\nonumber \\
&= 2\Pr{\I_t=0 \mid i\in D_t  \text{~and~} \F}\cdot \Pr{i\in D_t \text{~and~} \F} + n^{-2}\nonumber \\
&\leq 2 \Pr{\I_t=0 \mid i\in D_t \text{~and~} \F}\cdot \Pr{i\in D_t} + \frac{\beta}{n}\nonumber \\
&\leq \frac{2\beta d}{n}\, \Pr{\I_t=0 \mid i\in D_t \text{~and~} \F} + \frac{\beta}{n},\label{eq:intermediate}
\end{align}
using (\ref{eq:beta-balancedness}) for the final inequality.
Hence it suffices to prove that
		\begin{align}\label{u:cl}
		\Pr{\I_t=0 \mid i\in D_t \text{~and~} \F} \leq \hat{c}/d
		\end{align}
		for some absolute constant $\hat{c}>0$.
		From this, it follows from (\ref{ineq:prefirst}) and (\ref{eq:intermediate}) that $\Pr{A(t,i)} \leq \alpha/n$ where $\alpha = \beta(7 + 2\hat{c})$.
As $i$ was an arbitrary bin, this proves that the process is $(\alpha,m)$-uniform.
		Let $X$ be the random variable that counts the number of empty bins of a random $(d-1)$-element subset of  $H_t\setminus\{i\}$, 
		{conditioned on the event that ``$(i\in D_t)$ and $\F$'' holds.
			Then $X$} is a hypergeometric random variable with parameters $(\unif-1, K, d-1)$, 
		where $K$ is the number of empty bins contained in $H_t\setminus\{i\}$. 
		Thus 	\[
		\Ex{X}=\frac{(d-1)K}{\unif-1} \text{~~and~~} \Var{X}\le \frac{(d-1)K}{\unif-1}\le d.
		\]
		Then $\Ex{X} > d/3$, since $K\geq \unif/2-1$ when $i\in D_t$ and $\F$ holds
		(and using the size property $s = \Omega(\log n)$ and the fact that $d\geq 7$).
		Therefore
		\begin{align*}
		&	\Pr{\I_t=0 \mid (i\in D_t) \text{~and~} \F}\\
		&\le \Pr{X<d/6} \le \Pr{|X-\Ex{X}|\le \Ex{X}/2} < \frac{4\Var{X}}{\Ex{X}^2}
		\le\frac{4(d)}{(d/3)^2}=\frac{36}{d},
		\end{align*}
		using Chebychev's inequality. This establishes (\ref{u:cl}) with $\hat{c}=36$, completing the proof.
	\end{proof}

	We are ready to prove Lemma~\ref{lem:col}. 
	
	\begin{lemma}[Restatement of Lemma~\ref{lem:col}]\label{lem:col1} 
	Let $\beta,\, c_0\geq 1$ be constants and let 
				$\varepsilon = \varepsilon(n)\in (0,1)$ such
				that $n^\varepsilon\to\infty$.
		Let $d=d(n)$ be an integer-valued function such that $2\leq d = o(\varepsilon\, \log n)$.
		Let $(\H^{(1)},\ldots, \H^{(n)})$ be a dynamic hypergraph which satisfies the
			$\beta$-balanced,
			$\varepsilon$-visibility and $c_0$-size properties.
		Suppose that $c = c(\beta) \geq 44\beta \e^2$ is a sufficiently large constant,
		and let $k=C\log n$ for some constant
		$C\geq 1$. There exists $\Omega(n)= m\le n$ such that the probability that $\C_{m}$ contains a $c$-loaded  $k$-vertex tree  is at most 
		\[
		\exp\Big\{4k\log({2\beta d})- c(d-1)(k-r-1) + \big(c_0 + 3 - r\varepsilon/2\big) \log(n)\Big\}
		\]
		where $r$ is the number of red vertices in the  blue-red coloring of the tree. 
		%Moreover, with high probability, if $\C_{m}$ contains any such tree
		%then $r=\Oh(1/\varepsilon)$.
		Moreover, with probability $1-n^{-\omega(1)}$ no such tree exists with $\omega(1/\varepsilon)$ red vertices.
	\end{lemma}
	
	\begin{proof}
		Fix $m=m(n)$ to equal the $m$ provided by Lemma~\ref{lem:uni}.
		There are at most $4^k$ ordered trees with $k$ vertices.
		(Proposition~\ref{pro:ordered}). 
		Fix such a tree, say $T$, and label the vertices $\{1,2,\ldots, k\}$ 
		such that vertex $i$ is the $i$-th new vertex visited when performing depth-first search in $T$ starting from the root, and respecting the given ordering.
		In particular, the root of $T$ is vertex $1$. 
	Next, we will assign a $d$-choice to the root vertex of $T$, as a first
			step in describing trees which may be present in the 
			conflict graph $\mathcal{C}_m$.
		Let $x$ count the number of possible $d$-choices that can be assigned to the root of $T$. An upper bound on $x$ is given by the number of
		$d$-subsets of edges that have been present at some point
		during steps $t=1,\ldots, m$. Hence
		\[
		x\le \binom{\unif}{d}\cdot \left|\bigcup_{t=1}^{m}\Ec_t\right|\cdot m
		\le \binom{\unif}{d}\cdot n^{c_0+2},
		\]
		where the last inequality follows from the size property and the inequality $m\leq n$.   Therefore, there are $x$ possibilities for the root and hence  
			there are at most 
			$4^k\cdot \binom{\unif}{d}\cdot n^{c_0+2}$ ordered  trees with the specified root.  
			%Fix an arbitrary $d$-choice $D_t$ as the root for $T$.
		
		Next we fix an arbitrary function $\col:\{2,\ldots,k\}\rightarrow \{\text{blue},\text{red}\}$, that gives a  blue-red coloring of $2,\ldots,k$.  In what follows we  establish an upper bound for the probability that $\C_{m}$ contains the  blue-red colored   tree $T\subset \C_m$, (according to Definition~\ref{def:br}). 
		Let $q_1(t)$ be the probability that the $t$-th ball chooses the root of $T$
		{(that is, that the $d$-choice made by the $t$-th ball corresponds to the root of $T$)}. 
		Then
		\begin{align}\label{pr:no}
		\sum_{t=1}^{m}q_1(t)\le \sum_{t=1}^{m}\frac{1}{\binom{\unif}{d}}\le 
		\frac{n}{\binom{\unif}{d}}
		,
		\end{align}
		because $H$ contains  $\binom{\unif}{d}$ distinct $d$-element sets for every $H\in \Ec_t$.
		%%%
		For every $t=2,\ldots, k$, define $q_i(t,\col(i))$ to be  the  probability that 
		the $t$-th ball chooses the $i$-th vertex of the tree (i.e., $i$) with $\col(i)$. 
		If $\col(i)$ is red then $D_{t}$ must share at least two bins with  $\cup_{j=1}^{i-1}D_{t_j}$, 
		while if $\col(i)$ is blue then  $D_{t}$ only shares one bin with its parent. 
		For every  $i=2,\ldots, k$, let us derive an upper bound on $q_i(t, \text{blue})$.
		Here,  the $i$-th vertex  shares one bin with its parent in $T$, say $D_{t_j}$. 
		Now $D_{t_j}$ has $d$ bins and by the balancedness property  we get
		\[
		\Pr{D_{t_j}\cap H_t\neq \emptyset}\le\sum_{i\in D_{t_j}}\Pr{i\in H_t}\le  
		\frac{\beta d\unif}{n},
		\]
		where $H_t$ is the edge chosen by ball $t$ from $\H^{(t)}$, uniformly at random.
		Suppose that for some $a\ge 1$ we have  $|D_{t_j}\cap H_t|=a\le d$. 
		Then the total number  of $d$-element subsets of $H_t$ which share only one bin with $D_{t_j}$ is  $a\binom{\unif-a}{d-1}\le d\binom{\unif-1}{d-1}$. 
		Thus, we get
		\begin{align}\label{pr:blu}
		\sum_{t=1}^m q_i(t, \text{blue})\le \sum_{t=1}^m \frac{\beta d\unif}{n}\cdot
			d \frac{\binom{\unif-1}{d-1}}{\binom{\unif}{d}}=\sum_{t=1}^m\frac{\beta d^3}{n}
			\le \beta d^3,
		\end{align}
		because $m\le n$.
		
		Next, for every  $i=2,\ldots, k$, and every $t=2,\ldots, m$,
		we need an upper bound on $q_i(t, \text{red})$.
		If the $i$-th vertex of the tree is the set $D_t$ and is coloured red, then $D_t$
		is a $d$-element set of bins which shares at least two bins with
		$\cup_{j=1}^{i-1}D_{t_j}$.
		One of these bins belongs to the (known) parent, and the other belongs to $D_{t_1}\ldots,
		D_{t_{i-1}}$.  So if $U$ is the number of choices for this pair of bins, then
		\begin{align}\label{up:u}
		U\le d\cdot (i-1)d\le kd^2.
		\end{align}
		Let $\{p_1,p_2,\ldots, p_U\}$ be the set of such pairs of bins.

	As usual, let $H_t$ be the random edge chosen at step $t$. Then
			for any fixed pair of bins $\{x,y\}$,
				by the balancedness condition,
\begin{align}
%\Pr{p_J\subset  D_t}
& \hspace*{-1cm} \Pr{x\in D_t ~\text{and}~ y\in D_t}\nonumber \\
&= \Pr{x\in H_t ~\text{and}~ y\in H_t}\cdot \Pr{ x\in D_t ~\text{and}~ y\in D_t\mid x\in H_t ~\text{and}~ y\in H_t}  \nonumber \\
& = \Pr{x\in H_t \mid y\in H_t}\cdot \Pr{y\in H_t}\cdot  \frac{\binom{\unif-2}{d-2}}{\binom{\unif}{d}}\nonumber \\
&\le \Pr{x\in H_t \mid y\in H_t}\cdot\frac{\beta s}{n} \cdot \frac{d(d-1)}{\unif(\unif-1)}\nonumber \\
&=\Pr{x\in H_t\mid y\in H_t} \frac{\beta d(d-1)}{(\unif-1)n}\label{ineq:pairs}
\end{align}
		%as $\binom{\unif-2}{d-2}$ is the number of $d$-element subsets of $H_t$ which contain the pair $p_J$.  Then
We will apply this to $p_J = \{x,y\}$ for $J\in \{1,\ldots, U\}$.
By definition of visibility, $p_J$ appears at most $\viz{p_J}$ times as a subset of a hyperedge of the current hypergraph $\H_t$, for $t=1,\ldots, n$.
This is a property of hypergraphs exposed  over $n$ rounds.
Hence there exists a set of rounds, say
$L_J\subseteq \{1,\ldots, n\}$ with $|L_J|= n - \viz{p_J}$, such that for all $t\in L_J$, the  hypergraph 
$\H_t$ has no hyperedge which contains both $x$ and $y$. Therefore, for every 
$t\in L_J$ we have that  $\Pr{p_J\subseteq D_t} = 0$.
Thus, using (\ref{ineq:pairs}),
 \begin{align}\label{pr:red}
&\sum_{t=1}^{m} q_i(t, \text{red})\le \sum_{J=1}^U  
 \sum_{t=1}^{n} \Pr{p_J\subseteq D_t} \leq \sum_{J=1}^U \frac{\beta d(d-1)}{(\unif-1)n}\,\sum_{t \in [n]\setminus L_J}  \Pr{x\in H_t\mid y\in H_t}\nonumber\\
 &\le  \sum_{J=1}^U \frac{\beta d(d-1)}{(\unif-1)n} \viz{p_J}\leq 
 {\frac{\beta kd^4}{n^\varepsilon}}.
 \end{align}
The final inequality follows from (\ref{up:u}) the visibility property and the fact that $d\leq s$.

	Write $\col^{-1}(\text{blue})$ for the set of blue vertices in $T$,
		and similarly for $\col^{-1}(\text{red})$.  Then
		\[|\col^{-1}(\text{red})|+|\col^{-1}(\text{blue})|=k-1.\]
		Suppose that $(t_1,\ldots,t_k)$ is the sequence of balls that are going to select vertices $1,2,\ldots, k$ of $T$. By applying (\ref{pr:no}), (\ref{pr:blu}) and (\ref{pr:red}), we  find that the probability that the edges of the colored tree $T$ appears in $\C_{m}$ at times $(t_1,\ldots, t_k)$, and the corresponding sets
		$D_{t_1},\ldots, D_{t_k}$ consistent with the chosen blue-red coloring scheme,
		is at most 
		\begin{align}
		&\sum_{(t_1,\ldots,t_k)} \left\{q_1(t_1)\prod_{i=2}^k q_{i}(t_i, \col(i))\right\}
		\le \left(\sum_{t=1}^{m} q_1(t)\right){\prod_{i=2}^{k}}\left(\sum_{t=1}^{m} q_i(t, \col(i))\right)\nonumber\\
		&\leq \frac{n}{\binom{\unif}{d}} \,
		\left(  \prod_{i\in \col^{-1}(\text{blue})} \, 
		\sum_{t=1}^{m} q_i(t,\text{blue})\right) 
		\, \left(  \prod_{i\in \col^{-1}(\text{red})} \, 
		\sum_{t=1}^{m} q_i(t, \text{red})\right)
		\nonumber\\
		&\leq \frac{n}{\binom{\unif}{d}} \left({\beta d^3}\right)^{|\col^{-1}({\text{blue}})|} \left({\frac{\beta kd^4}{n^{\varepsilon}}}\right)^{|\col^{-1}(\text{red})|}
		\le \frac{{n\beta^k d^{4k}}}{\binom{\unif}{d}}  \left(\frac{k}{n^{\varepsilon}}\right)^{|\col^{-1}(\text{red})|}. \label{huge}
		\end{align}  
		There are at most $2^{k-1}$ coloring functions and $4^{k}\text{poly}(n) \binom{\unif}{d}$ rooted and  ordered trees. So by the upper bound (\ref{huge}),
		together with the union bound over all colored ordered trees, we obtain
		\begin{align}\label{ineq:final}
		&\Pr{\text{ $\C_{m}$ contains a valid blue-red colored $k$-vertex tree with $r$ red vertices }}\nonumber\\
		&\le  4^k 2^{k-1}\cdot n^{c_0+2} \binom{\unif}{d} \cdot \frac{{n\, \beta^k d^{4k}}}{\binom{\unif}{d}}  \left(\frac{k}{n^{\varepsilon}}\right)^{r}\nonumber \\
		& \le n^{c_0+3}\cdot (2\beta d)^{4k} \cdot n^{-r\varepsilon/2}
		%\le \exp\big( n^{c_0+2}\e^{4k\log(2d) + (c_0+2 -r\varepsilon/2} )\log n\big).
		\le \exp\big( 4k\log(2\beta d) + (c_0+3 -r\varepsilon/2 )\log n\big),
		\end{align}
		using $k=\Oh(\log n)$ for the penultimate inequality.
		
		Let $b=k-r-1$ be the number of blue vertices and let
		$D_{s_1},\ldots,D_{s_{b}}$ be the sorted list of blue vertices such that
		$s_1<s_{2} < \cdots < s_b$.   
		Then, by the definition of blue-red coloring, for every $j=1,\ldots, b$ we have 
		$
		|(\cup_{g=1}^{j-1}D_{s_g}) \cap D_{s_j}|\le 1$. 
		%because each blue vertex in
		%the tree shares at most one bin with the previous blue $D_t$'s in the sorted list. 
		This implies that 
		\[
		y=|\cup_{j=1}^{k} D_{t_j}|\ge|\cup_{j=1}^{b} D_{s_j}|\ge (d-1)b=(d-1)(k-1-r), 
		\]
		{since $\{s_1,\ldots, s_b\}\subseteq \{t_1,\ldots, t_k\}$}.
		Applying Lemma~\ref{lem:uni} implies that the 
		balanced allocation is $(\alpha, m)$-uniform, where $\alpha = 44\beta$, say. 
		Hence for any $c\geq 44\beta \e^2$, the  probability that each bin in
		$\cup_{j=1}^{k} D_{t_j}$ is allocated at least $c$ balls 
		(that is, the tree $T$ is $c$-loaded) is at most 
		\[
		\binom{m}{cy} \left(\frac{\alpha y}{n}\right)^{cy}\le 
		\left(\frac{\e m}{cy}\right)^{cy} \left(\frac{\alpha y}{n}\right)^{cy}\le \left(\frac{\e\alpha}{c}\right)^{cy}\le \e ^{-c(d-1)(k-r-1)},
		\]
		where the last inequality follows from $m\le n$ and the fact that 
		$c>\alpha \e^2$. Since balls are independent from each other, we can multiply the above 
		inequality by  (\ref{ineq:final}) to show that the probability that
			$\C_m$ contains a $c$-loaded $k$-vertex tree with $r$ red vertices is at most
		\begin{equation}
		\label{eq:prob-bound}
		\exp\Big\{4k\log(2\beta d)- c(d-1)(k-r-1) + \big(c_0 + 3 - r\varepsilon/2\big) \log n\Big\},
		\end{equation}
		proving the first statement of the lemma.
		%&\le \exp\{c_0\log n+4k\log(2d)-r\varepsilon \log n/2-c(d-1)(k-r-1)\},\\
		%&\le \exp\{c_0\log n+c(d-1)+k(4\log(2d)-c(d-1))+ r(c(d-1)-\varepsilon \log n/2)\}\\
		%&\le \exp\{\Theta(\log n)-\Theta(kd) -r\varepsilon\Theta(\log n)\}
		%\end{align*}
		%which holds because there exists some $c_0$ such that  $\text{poly}(n)\le n^{c_0}$ and $d=o(\log n)$. 
			Finally, suppose that $r\varepsilon\to\infty$ as $n\to \infty$.  Recall that $\beta$ and $c=c(\beta)$ are positive constants,
			$d=o(\varepsilon\, \log n)$ and $k=\Theta(\log n)$.
			If $d=O(1)$ then $4\log(2\beta d) - c(d-1) = O(1)$,
			and otherwise $4\log(2\beta d) - c(d-1) < 0$
			for $n$ sufficiently large.
			Hence the upper bound
			in (\ref{eq:prob-bound}) can be written as
			\begin{align*}
			& \exp\Big\{ \big( 4\log(2\beta d) - c(d-1) \big) k + c(d-1)(r-1) - (r\varepsilon/2)\log n  + \Oh(\log n) \Big\}\\
		&\leq \exp\Big\{ cdr - (r\varepsilon/2)\log n + \Oh(\log n)\Big\}\\
		&\leq \exp\Big\{  - \Theta(r\varepsilon)\log n\Big\}
			\end{align*}
		since $r\varepsilon\to\infty$ and $d=o(\varepsilon\, \log n)$.  
	This probability is $n^{-\omega(1)}$ even after multiplying by $n$ to cover
		the possible values for $r$.
	This completes the proof.
	\end{proof}

	\section{Proof of Theorem \ref{thm:lower-bound}}\label{sec:lower-bound}
	\begin{proof}
			Let $G=([n],E)$ denote an $s$-regular graph that does not contain any $4$-cycle, where $s= n^{\varepsilon}$.
			It is worth mentioning that there are several explicit families of $s$-regular graphs with  girth $\log _s n$ (e.g., see \cite{Dahan14}). For each $i\in [n]$, let $N(i)$ be the set of vertices adjacent to $i$. Also, let  
			$\H=([n], \{N(i) \mid i=1,\ldots,n\})$ denote a hypergraph obtained from $G$. 
			We consider the $s$-uniform dynamic hypergraph $(\H,\H,\ldots,\H)$.
			%is the $s$-uniform dynamic hypergraph.
			Clearly, for every $\{i,j\}\subset[n]$ we have that 
			\[
			\viz{i,j}\le n\le s n^{1-\varepsilon}
			\]
			Therefore, the dynamic hypergraph is $\varepsilon$-visible.
			{Fix an integer $d$ such that $2\leq d\leq s$ and $d$ is constant.}
			Since $G$ does not contain any $4$-cycle, we deduce that every $d$-{subset} of vertices only appears in at most one hyperedge of $\H$. Therefore, the probability that a $d$-{subset} is chosen by any ball is $1/(n{ s\choose d})$. 
			Let $D=\{i_1,i_2,\ldots,i_d\}\subset [n]$ be an arbitrary  set of $d$ vertices contained in some hyperedge of $\H$. Let $X(D,k)$ be an indicator random variable taking one if at least $k$ balls choose $D$ and zero otherwise. Then we have that   
			\[
			\Pr{X(D,k)=1}={n \choose k}\left(\frac{1}{n{s \choose d}}\right)^k
			\]
			Also let $Y_k=\sum_{D}X(D,k)$ denote the number of $d$-{subsets} that are chosen by at least $k$ balls. By linearity of expectation we have that 
			\begin{align}\label{eq:exp}
			\Ex{Y_k}=\sum_{D}\Ex{X(D,k)}=n{s \choose d}{n \choose k}\left(\frac{1}{n{s \choose d}}\right)^k\ge n\left(\frac{ s^{-d}}{k}\right)^k=n\left(\frac{ n^{-d\varepsilon}}{k}\right)^k,
			\end{align}
			where the last inequality follows from ${n \choose k}\ge (\frac{n}{k})^k$ and ${s\choose d}<s^{d}$. In what follows we show that with high probability there exists $k$ such that $Y_k\ge 1$.
			Suppose that $d\varepsilon=\Theta(1)$, then if we set $k=1$, then  there is a $d$-subset which is picked by at least one ball and hence $Y_1\ge 1$.  If $(\log\log n)/(3\log n)<d\varepsilon$ and $d\varepsilon=o(1)$, then by setting   $k=1/(6d\varepsilon)$ we have $k<(\log n)/(2\log\log n)<\log n$ and 
			\[
			\Ex{Y_k}\ge n k^{-k}n^{-kd\varepsilon}\ge n (\log n)^{-\log n/(2\log \log n)} n^{-1/6}={n^{1/3}}=\omega(\log n).
			\]
			Moreover, if $d\varepsilon \le \log\log n/(3\log n)$, then by letting  $k=\log n/(2\log\log n$) we get that 
			\[
			\Ex{Y_k}\ge n k^{-k} n^{-kd\varepsilon}\ge n (\log n)^{-\log n/(2\log\log n)}n^{-1/6}=n^{1/3}=\omega(\log n).
			\] 
			Therefore, there exists $k=\min\{\Omega(1/\varepsilon), \Omega(\log n/\log\log n)\}$ so that $\Ex {Y_k}=\omega(\log n)$.	
			As the number of balls is $n$,   for a given $k$, the random variables $X(D,k)$  are negatively correlated.
			This follows because if  some set  $D$ is chosen at least $k$ times, then the number of balls which may choose some other set  reduces to at most  $n-k$, and hence   the chance  for other sets to be  chosen at least $k$ times is decreased.
			 Application of the Chernoff  bound for negatively correlated random variable implies that 
			\begin{align*}
			\Pr{Y_k\le \Ex{Y_k}/2}\le \exp({-\Ex{Y_k}}/{8})=\exp(-\omega(\log n)) = n^{-\omega(1)}. 
			\end{align*}
			It follows that with probability $1-n^{-\omega(1)}$ there exists a $d$-subset $D$ which is chosen by at least $k$ balls, and hence with probability $1-n^{-\omega(1)}$ there is at least one bin in $D$ whose load is at least $k/d$.       
		\end{proof}

	\section{Balanced Allocation on Dynamic Graphs}\label{sec:graph}
	
	In this section we show an upper bound for maximum load attained by the balanced allocation on regular dynamic graphs (i.e.,  Theorem~\ref{thm:s2c}).
Our arguments %The lemmas and their proofs 
are inspired by~\cite[Lemma 2.1 and 2.2]{KP06}.
	%The algorithm proceeds in successive rounds, in each round $1\le t\le n$ ball $t$ picks a random edge from $G^{(t)}$, where $G^{(t)}$ is a $\Delta_t$-regular graph whose vertex set is the set of $n$ bins. Then the ball is allocated to an endpoint of the edge with a lesser load, with ties are broken randomly. 

	Suppose that the balanced allocation process has allocated $n$ balls to the dynamic
	regular graph $(G^{(1)},\ldots, G^{(n)})$.
	Define the  \emph{conflict graph} $\C_n$ formed by the edges selected by the $n$ balls.
	The vertex set of $\C_n$ is the set $[n]$ of bins, and the loads of these bins are updated during the process.
	
Given a tree $T$ which is a subgraph of $\C_n$, and vertices $u$, $v$ of the tree,
if $\{u,v\}$ is an edge of $\C_n$ then we say it is a \emph{cycle-producing edge}
with respect to the tree $T$.  The name arises as adding this edge to the tree
would produce a cycle, which may be a 2-cycle if the edge $\{u,v\}$ is already 
present in $T$.
	%and a new edge, if the both endpoints of the edge is contained in the tree, then the edge is called cycle-producing edge, which also includes the multiple-edges.  
	For a positive integer $c>0$, a  subgraph of $\C_n$ is called $c$-\emph{loaded} if each vertex (bin) contained in the subgraph has load at least $c$.

The following  proposition presents some properties of connected components of $\C_n$.
	
	\begin{proposition}\label{pro:dy} 
		Let $(G^{(1)}, \ldots, G^{(n)})$ be a regular dynamic graph on vertex set $[n]$
		which is $\varepsilon$-visible.  Let that $\C_n$ be the conflict graph obtained after 
		allocating $n$ balls using the balanced allocation process. Then
		for every given constant $c>0$, with probability at least $1-n^{-8c}$,  every $12(c+1)$-loaded connected component of $\C_n$ contains strictly fewer than $\log n$ vertices. Moreover, with probability at least $1-n^{-12(c+1)}$, the number of cycle-producing edges in each such component is at most $2(12c+13)/\varepsilon$.  % Was $2(c+1)/\varepsilon$ but should be $2(c_1+1)/\varepsilon$, with $c_1 = 12(c+1)$.
	\end{proposition} 

	We will prove the proposition in Section~\ref{sub:dy}. 
We now explain how to recursively build a witness graph, provided there exists a bin whose load is higher than a certain threshold.
	
	\paragraph*{Construction of the Witness Graph}{ Let us start with a bin, say $r$, with  $\ell+c$ balls. Clearly, if a ball is in bin $r$ at height  $h$ then the other bin it chose, as part of the balanced allocation procedure, had load at least $h$.
		Starting from bin (vertex) $r$, let us recover all $\ell$ edges corresponding to the balls that were placed in $r$ with height at least $c$. Thus, the alternative bin choices have loads at least $\ell+c-1,\ldots,c$, respectively.  These $\ell$ bins are all neighbours of 
		$r$ in $\C_n$, and we refer to them as the children of $r$.  Next, we recover the edges corresponding to balls placed in the children of $r$ at height at least $c$. Recursively, we continue until there is no ball remaining at height $c$ or more. For every $i=1,\ldots, \ell$, 
		let $f(\ell-i)$ denote the number of vertices generated by the recursive construction,
		starting with a bin which contains $\ell-i+c$ balls.
		Assume for the moment that, for each vertex with load at least $c$, the recursive procedure always produces distinct children. 
		Then
		\[
		f(\ell)\ge f(\ell-1)+f(\ell-2)+\ldots+f(0)+1,
		\]
		where $f(0)=1$. 
		A simple calculation shows that  $f(\ell)\ge 2^\ell$. 
		Thus, the recursive procedure gives a  $c$-loaded tree with at least $2^{\ell}$ vertices, 
		{under the assumption that the children of each vertex considered by the recursion are
			all distinct.}
	}
	\medskip
	
	We may now prove our main result on dynamic regular graphs.
	
	\begin{proof}[Proof of Theorem~\ref{thm:s2c}]\
		We want to show that after $n$ balls have
		been allocated to the dynamic regular graph $(G^{(1)},\ldots, G^{(n)})$,
		which satisfies the $\varepsilon$-visibility property,
		the maximum load is at most $\log_2 \log n + \Oh(1/\varepsilon)$ with
		high probability.
		
		Let $c>0$ be a given constant.
		By the second statement of Proposition~\ref{pro:dy}, 
with probability at least $1-n^{-12(c+1)}$,
		the number of cycle-producing edges in a given
			component of $\C_n$ is at most $c_2 = 2(12c+13)/\varepsilon$. 
		For a contradiction, suppose that there exists a bin, say $r$, 
		which has at least $\ell+c_1+c_2+1$ balls, where $c_1 = 12(c+1)$.
		Consider $c_2+1$ balls in $r$ at height at least $\ell+c_1$. 
		The children of $r$ in $\C_n$ are the bins $r_1,r_2,\ldots, r_{c_2+1}$ 
		(which might not be distinct), which were the alternative choice of these $c_2+1$ balls.
		Each of these children $r_i$ has load at least $\ell + c_1$. 
		We start the recursive construction at each child $r_i$ of $r$.  Assuming that
		this component of $\C_n$ contains at most $c_2$ cycle-producing edges, it follows that 
		for at least one
		child $r_i$ of $r$, the recursive procedure gives distinct children for each vertex
		which is a descendent of $r_i$. Hence we obtain a $c_1$-loaded tree which has
		$2^{\ell}$ vertices. 
		Substituting $\ell = \log_2\log n$ and applying the first statement of Proposition~\ref{pro:dy}, 
we conclude that with probability at least $1-n^{-8c}$
		such a structure does not exist in $\C_n$.  
		This contradiction shows that %with high probability, 
with probability at least $1-2n^{-8c}\geq 1-n^{-c}$,
the maximum load after
		$n$ balls have been allocated is at most $\log_2\log n +\Oh(1/\varepsilon)$. 
	\end{proof}

\subsection{Proof of Proposition~\ref{pro:dy}}\label{sub:dy}

In this subsection we will prove two lemmas and then combine them to establish Proposition~\ref{pro:dy}. %the proposition. 
	%Recall that a subgraph of $\C_n$ is $c$-loaded if every vertex (bin) in the subgraph has load at least $c$.
	
	\begin{lemma}\label{lem:tree}
	Let $k$ be a positive integer and let $c_1>0$.
		The probability that conflict graph $\C_n$ contains a $c_1$-loaded
		connected component with $k$ vertices is at most 
		\[
		n\cdot 8^{k}\cdot\left(\frac{2\e}{c_1}\right)^{c_1k}.
		\]
		Moreover, by setting $c_1=12(c+1)$, we conclude that with probability at least $1-n^{-8c}$, the conflict graph $\C_n$ does not contain a $c_1$-loaded tree with at least
		$\log n$ vertices.
	\end{lemma}

	\begin{proof}
		A connected component in $\C_n$ with $k$ vertices contains a spanning tree with $k$ vertices. By Proposition~\ref{pro:ordered}, there are at most $4^{k-1}$ ordered trees
		{with $k$ vertices}. For every ordered tree, we can choose its root in $n$ ways, as we have $n$ bins (vertices).  
		Hence there are at most $n\cdot 4^{k-1}$ rooted and ordered trees. Let us fix an arbitrary ordered tree $T$ with a specified root. Also let $(t_1, \ldots, t_{k-1})$ denote an arbitrary sequence of rounds, where $t_i \in\{1,\ldots, n\}$ is the round when the $i$-th edge of the ordered tree  $T$ is chosen. 
		Notice that in an  ordered tree with specified root, the $i$-th edge always connects the $i$-th child to its parent, and the parent is already known to us. Therefore, to build the tree, the $i$-th edge of the tree must be chosen from edges of $G^{(t_i)}$ that are adjacent to the known parent. This implies that the algorithm chooses the $i$-th edge of $T$ in round $t_i$ with probability 
		$\frac{\Delta_{t_i}}{n\Delta_{t_i}/2}=\frac{2}{n}$.
		Since  balls are independent from each other,  the tree $T$  is constructed at
		the given times $(t_1,\ldots,t_{k-1})$ with probability 
		\begin{align}\label{up:1}
		\left(\frac{2}{n}\right)^{k-1}.
		\end{align} 
		On the other hand, ball $t$ is allocated to a given bin with probability at most $\Delta_t/(n\Delta_t/2)=2/n$. Therefore, the probability that $T$  is $c_1$-loaded is at most
		\begin{align}\label{up:2}
		\binom{n}{ck}\left(\frac{2k}{n}\right)^{c_1k}\le \left(\frac{\e n}{c_1k}\right)^{c_1k}\left(\frac{2k}{n}\right)^{c_1k}=\left(\frac{2\e}{c_1}\right)^{c_1 k},
		\end{align}
		where we used the fact that $\binom{n}{c_1k}\le \left(\frac{\e n}{c_1k}\right)^{c_1k} $. Since balls are independent, one can multiply (\ref{up:1}) by (\ref{up:2}) and derive an upper bound for the probability that $T$  is constructed at the given times and 
		is $c$-loaded.
		Taking the union bound over all rooted ordered trees and time sequences gives
		\begin{align*}
		n4^{k-1}\sum_{(t_1,\ldots, t_{k-1})}\left\{\left(\frac{2}{n}\right)^{k-1}\left(\frac{2\e}{c_1}\right)^{c_1k}\right\}
		&\le n4^{k-1} n^{k-1}\cdot \left\{\left(\frac{2}{n}\right)^{k-1}\left(\frac{2\e}{c_1}\right)^{c_1k}\right\}\\
		&=n8^{k-1}\left(\frac{2\e}{c_1}\right)^{c_1k}, 
		\end{align*}
		proving the first statement of the lemma.
By setting $c_1=12(c+1) > 4\e$ and $k=\log n$ in the above formula, we infer that
		the probability that $\C_n$ contains a $c_1$-loaded tree with $\log n$ vertices is at most
		\[
		n8^{k-1}\left(\frac{2\e}{c_1}\right)^{c_1k}< n2^{3k}2^{-12(c+1)k}\le 
n2^{-12ck-9k}\le n^{-8c},
		\] 
		completing the proof.
	\end{proof}

	\begin{lemma}\label{lem:treecycle}
Suppose that the assumptions of Proposition~\ref{pro:dy} hold.
		Further assume %Suppose 
that the conflict graph $\C_n$ contains a $c_1$-loaded $k$-vertex tree $T$, where $c_1> 4\e$ is any constant  and $k$ is a positive integer.
		Let $p$ denote the number of cycle-producing edges (with respect to $T$) 
			which have been added
			between vertices in this tree during the allocation process.
		Then $p < 2(c_1+1)/\varepsilon$ with probability at least  $1-n^{-c_1}$. 
	\end{lemma}

	\begin{proof}
	For a given connected component of $k$ vertices, there are at most $\binom{k}{2}$ 
	edges whose addition may produce a cycle. This includes edges already present in the
		component, as an edge with multiplicity 2 (double edge) forms a 2-cycle.
	Thus, the $p$ edges can be chosen in $\binom{k}{2}^p<k^{2p}$ ways. Let $\{e_1,e_2,\ldots, e_p\}$ denote a set of $p$ cycle-producing edges (some of these may create 2-cycles).
	Also let $(t_1,\ldots, t_p)$ denote a sequence of rounds, where 
	$t_i\in\{1,\ldots, n\}$ is the round in which the $t_i$-th ball picks edge $e_i$. For each round $t=1,2,\ldots, n$ and $i=1,\ldots, p$, let us define $\I_t(e_i)$ as follows:
		\[
		\I_t(e_i)= \begin{cases} 
		1 &  \text{ if $e_i\in E_t$,}\\
		0 &  \text{ otherwise.}
		\end{cases}
		\] 
%to equal~1 if $e_i\in E_t$, and~0 otherwise.
		It is easy to see that 
		\[
		\Pr{\text{ball $t$ picks edge $e_i $ of $G^{(t)}$}}=\frac{\I_t(e_i)}{|E_t|}.
		\]
		Now $\viz{e_i}=\sum_{t=1}^n \I_t(e_i)$ for $i=1,\ldots, p$.
		Using this, and the fact that $|E_t|\geq n/2$ for each $t$  (since $G^{(t)}$ is regular with degree at most 1), 
		the probability that  $e_1,e_2,\ldots, e_p$ are chosen is at most 
		\begin{align}\label{in:cyc}
		\sum_{(t_1,\ldots,t_p)}\left\{\prod_{i=1}^p 
		\frac{\I_{t_i}(e_i)}{E_{t_i}}\right\}
		\le \prod_{i=1}^p\left\{\sum_{t=1}^{n}\frac{\I_{t}(e_i)}{E_{t}}\right\}
		\le \prod_{i=1}^p\frac{{4\viz{e_i}}}{n}\le \left(\frac{{4n^{1-\varepsilon}}}{n}\right)^p=\left(\frac{{4}}{ n^{\varepsilon}}\right)^p.
		\end{align}
		Moreover, applying Lemma~\ref{lem:tree} shows that the probability 
			that $\C_n$ contains a $c_1$-loaded $k$-vertex tree is at most
		\begin{equation}
		\label{new}
		n\cdot 8^k\cdot \left(\frac{2\e}{c_1}\right)^{c_1k}\le n\cdot 2^{-k},
		\end{equation}
		as $c_1 \geq 4\e$.
		So, with high probability, $\C_n$ does not contain any $c_1$-loaded tree with at
		least $(\log n)^2$ vertices. 
		Now assume that $k < (\log n)^2$.  Combining (\ref{in:cyc}) and (\ref{new}),
		and taking the union bound
		over all choices for a set of $p$ edges, we find that the probability that a $c_1$-loaded $k$-vertex tree contains $p$ cycle-producing edges is at most  
		\begin{align}\label{ineq:cycle}
		k^{2p}\cdot \left(\frac{{4}}{ n^{\varepsilon}}\right)^p\cdot n\cdot 2^{-k}= 
		\left(\frac{{4\cdot k^{2}}}{ n^{\varepsilon}}\right)^p\cdot n\cdot 2^{-k}\le n^{-\varepsilon p/2} \cdot n\cdot 2^{-k},
		\end{align}
		where the inequality holds as $k<(\log n)^2$.
		Therefore the probability that $p=  \lceil 2(c_1+1)/\varepsilon\rceil$ cycle-producing edges
		are present is at most $n^{-c_1}$.
		We conclude that $p < 2(c_1+1)/\varepsilon$  with probability at least $1-n^{-c_1}$. 
	\end{proof}

	\begin{proof}[Proof of Proposition~\ref{pro:dy}]
Combining the Lemmas~\ref{lem:tree} and~\ref{lem:treecycle} establishes the 
proposition, using $c_1 = 12(c+1)$ in Lemma~\ref{lem:treecycle}.
	\end{proof}

	\section{Proof of Proposition~\ref{pro:gmn}}\label{sec:pop}

We conclude the paper with the deferred proof of Proposition~\ref{pro:gmn}. First we restate a useful theorem from~\cite{CLLM12}.
	
	\begin{theorem}\emph{\cite[Theorem 3]{CLLM12}}\ \label{chernof}
		Let $M$ be an ergodic Markov chain with finite state space $\Omega$ and stationary
		distribution $\pi$. Let $T = T(\varepsilon)$ be its $\varepsilon$-mixing time for $\varepsilon<1/8$. Let $(Z_1,\ldots, Z_t)$ denote a
		$t$-step random walk on $M$ starting from an initial distribution $\rho$ on
		$\Omega$ (that is, $Z_1$ is distributed according to $\rho$). For
		some positive constant $\mu$ and
		every $i \in [t]$, let $f_i:\Omega\to [0, 1]$ be a weight function at step $i$ such that the expected weight 
		$\Expi{f_i(v)}{\pi} = \sum_{v\in\Omega} \pi(v) f_i(v)$ satisfies
		$\Expi{f_i(v)}{\pi} = \mu$ for all $i$. 
		Define the total weight of the walk $(Z_1, . . . , Z_t)$ by $X=\sum_{i=1}^t f_i(Z_i)$.
		Write $||\rho||_{\pi}=\sqrt{\sum_{x\in \Omega} \rho_x^2/\pi_x}$. 
		Then there exists some positive constant $c$ $($independent of $\mu$ and 
		$\varepsilon)$ such that for all $\alpha\geq 0$,
		\smallskip
		\begin{enumerate}
			\item   $\Pr{X\ge (1+\alpha)\mu t}\le c||\rho||_\pi\, \e^{-\alpha^2\mu t/72 T}$ \,\,\, for $0\le \alpha\le 1$. 
			\item   $\Pr{X\ge (1+\alpha)\mu t}\le c||\rho||_\pi \,\e^{-\alpha\mu t/72 T}$\,\, \,\,\, for $\alpha>1$.
			\item   $\Pr{X\le (1-\alpha)\mu t}\le c||\rho||_\pi \,\e^{-\alpha^2\mu t/72 T}$ \,\,\, for $0\le \alpha\le 1$.
		\end{enumerate}
	\end{theorem}
	
	\begin{proof}[Proof of Proposition~\ref{pro:gmn}] 
		Let $\Omega$ be the vertex set of the $R$-dimensional torus $\Gamma(n, R)$ and let
		$a$ and $b$ denote two arbitrary agents. By definition of the communication graph process, agents $a$ and $b$ are initially placed on two randomly chosen vertices of $\Gamma$, say $u_0$ and $v_0$.  Note that  $u_0$ and $v_0$ are independently chosen according to the stationary distribution  $\pi$ of the random walk on $\Gamma(n, R)$.
		Now consider the trajectory of agents $a$ and $b$, which give two independent random walks 
		$u_0, u_1,\ldots $ and $v_0, v_1,\ldots$  on $\Gamma(n, R)$.
		Defining $X_t=(u_t, v_t)$ for  $t=0,1,\ldots$ gives a finite, ergodic Markov chain with 
		stationary distribution $(\pi,\pi)$ on $\Omega\times\Omega$.
		For every $t\ge 0$, define
$f(X_t) = f(u_t,v_t)$ to equal~1 if $d_t(u_t,v_t)\leq r$, and equal~0 otherwise,
		where $d_t(\cdot,\cdot)$ is the Manhattan distance for the given grid. 	
		Let $u^1_t$ and $v^1_t$ denote the projection of the random walks $u_t$ and $v_t$ onto the $1$-dimensional torus $\Gamma(n^{1/R}, 1)$, respectively, defined by taking the first component of each of the random walks on $\Gamma(n,R)$. Then $X^1_t=(u^1_t, v^1_t)$ is  an ergodic Markov chain on  $\Gamma(n^{1/R},1)$, and its initial distribution is stationary.  
		We may also define 
$f(u^1_t,v^1_t)$ to be~1 if $d_t(u^1_t,v^1_t)\leq r$, and~0 otherwise.
		By the Manhattan distance property, if $f(u_t,v_t)=1$ then $f(u^1_t,v^1_t)=1$. Therefore,
		\[
		\viz{a,b}=\sum_{t=0}^nf(X_t)\le \sum_{t=0}^nf(X^1_t).
		\]
		Set $\delta=\min\{1/4, 1/R\}$.  Let $t_0$ be the first time when $d_{t_0}(u^1_{t_0}, v^1_{t_0})\le n^{\delta}$. 
		Consider a moving window $W$ of length  $2n^{\delta}+1$, which contains the locations of $u^1_{t_0}$ and $v^1_{t_0}$. 
		At time $t_0$, the vertices covered by $W$ are labelled in increasing order, with the leftmost vertex labelled~{$-n^{\delta}$}
		and the rightmost vertex labelled $n^{\delta}-1$.  The window $W$ stays at its initial location as long as no agent hits a border of $W$ (vertices labelled $-n^{\delta}$ or $n^\delta$), or the middle vertex of $W$ (labelled 0). 
		Let $b$ be the first agent that hits a border or the centre of $W$. From this time on, $b$ and $W$  are coupled so that
		they both move and/or stay, simultaneously. (If $b$ moves left then $W$ also moves left, for example.) 
		Each time the window $W$ moves, a vertex $u\in \Gamma_1$ is no longer covered by $W$ and a new vertex, $w\in\Gamma_1$, 
		becomes covered by $w$. 
		The new vertex $w$ is assigned the label of vertex $u$. 
		This process always  labels the vertices covered by $W$ by $\{-n^{\delta-1},\ldots,n^{\delta}-1\}$,
		and the movement of agent $b$ over these labeled vertices  simulates a random walk on the additive group  
		$\mathbb{Z}_{2n^\delta + 1}$.
		Define 
		\[
		S=\{1\le t\le n \mid u^1_t ~\text{and}~ v^1_t \in W \}.
		\]
		Assume that $S\neq \emptyset$ and define the chain $Y_t=(u^1_t, v^1_t), t\in S$.
		Then $Y_t$ can be considered as an ergodic Markov chain of length $|S|\le n$  over $\mathbb{Z}_{2n^\delta-1}$,
		or equivalently, as a Markov chain on a $(2n^{\delta}+1)$-cycle. By the proposition assumption we have 
		$r= \Oh(n^{o(1)}) <n^{\delta}$, and so
		\[
		\viz{a,b}=\sum_{t=0}^nf(X_t)\le \sum_{t=0}^nf(X^1_t)\leq \sum_{t\in S}f(Y_t)\le \sum_{t=0}^nf(Y_t).
		\]
		The chain $Y_t$ converges to stationary distribution $(\pi, \pi)$, where $\pi$ is the uniform
		distribution of a random walk on a $(2n^{\delta}+1)$-cycle.
		It follows that for all $t =0,1,\ldots$ we have
		$\Expi{f(Y_t)}{(\pi,\pi)} = \mu = \Theta(r/n^\delta)$, independently of $t$. It is well-known~\cite{LPW06} that 
		the $\varepsilon$-mixing time of the random walk on  a $(2n^\delta+1)$-cycle is $\Oh(n^{2\delta}\log(1/\varepsilon))$. If $\rho$ is the initial distribution $Y_0$, then we have that 
		$||\rho||_\pi\le \Oh(n^{\delta})$.
		Applying Theorem~\ref{chernof} implies that
		\[
		\Pr{\sum_{t=1}^{n}f(Y_t)\ge  \mu \cdot n}
		=\Oh(n^{\delta})\e^{-\Theta(rn^{1-3\delta})}=n^{-\omega(1)}.
		\]
		Therefore, with probability $1-n^{-\omega(1)}$, 
		\[
		\viz{a,b}\le \sum_{t=0}^nf(Y_t) =  \Oh(rn^{1-\delta})= \Oh(n^{1-\delta + o(1)}) = \Oh(n^{1-\varepsilon}),
		\] 
		taking $\varepsilon = \delta/2$, say.
		Taking the union bound over all pairs of agents completes the proof.

	\end{proof}

	%%%%%%%%%%%%%%%%

\end{document}